%% file: ms.tex
\documentclass{llncs}
\usepackage[utf8]{inputenc}
\usepackage{qip}
\usepackage{framed}             
\usepackage[dvipsnames]{xcolor}
\usepackage{bbm}                
\usepackage{hyperref}           
\usepackage[backend=biber,style=numeric,maxcitenames=4,maxalphanames=100,maxbibnames=100,isbn=false]{biblatex}

\DeclareFieldFormat{eprint:iacr}{%
  IACR eprint\addcolon\space
  \ifhyperref
    {\href{http://eprint.iacr.org/#1}{\texttt{#1}}}
    {\texttt{#1}}}

\DeclareMathOperator{\spn}{span}
\DeclareMathOperator{\sym}{Sym}

\newcommand{\idcptp}{\mathrm{id}}
\renewcommand{\id}{\mathbbm{1}}
\newcommand{\ident}{\mathbbm{1}}

\newcommand{\symgrp}[1]{{\cal S}_{#1}}
\newcommand{\Eacc}{{\cal E}^\acc_{RS^N\rightarrow S^n}}
\newcommand{\barEacc}{\bar{\cal E}^\acc_{P^N S^N\rightarrow \Pi S^n}}
\newcommand{\rhobar}{\bar \rho_{P^NS^N}}

\newcommand{\eprlocc}{{\bf EPR-LOCC Sampling }}
\newcommand{\purifbased}{{\bf Purification-Based Sampling }}
\newcommand{\generalsampling}{{\bf General Mixed State Sampling
    Protocol }}

\newcommand{\cptn}{completely positive trace non-increasing }

\newcommand{\acc}{\mathrm{acc}}

\newcommand{\Span}{\mathrm{span}}

 \title{Secure Certification of Mixed Quantum States with Application to
   Two-Party Randomness Generation} 
\author{}\institute{}
\author{Frédéric
   Dupuis\inst{2,3}\and Serge Fehr\inst{1}\and Philippe
   Lamontagne\inst{4}\and Louis Salvail\inst{4}
 }

 \institute{CWI, Amsterdam, The Netherlands \and
     Université de Lorraine, CNRS, Inria, LORIA, F-54000 Nancy, France \and
     Faculty of Informatics, Masaryk University, Brno, Czech Republic \and
     Université de Montréal (DIRO), Montréal, Canada
   }

\begin{document}
\maketitle
\begin{abstract}
We investigate sampling procedures that certify that an arbitrary quantum state on $n$ subsystems is close to an ideal mixed state $\varphi^{\otimes n}$ for a given reference state $\varphi$, up to errors on a few positions. This task makes no sense classically: it would correspond to certifying that a given bitstring was generated according to some desired probability distribution. However, in the quantum case, this is possible if one has access to a prover who can supply a purification of the mixed state. 

In this work, we introduce the concept of mixed-state certification, and we show that a natural sampling protocol offers secure certification in the presence of a possibly dishonest prover: if the verifier accepts then he can be almost certain that the state in question has been correctly prepared, up to a small number of errors. 

We then apply this result to two-party quantum coin-tossing. Given
that strong coin tossing is impossible, it is natural to ask ``how
close can we get". This question has been well studied and is nowadays
well understood from the perspective of the bias of individual coin
tosses. We approach and answer this question from a different---and
somewhat orthogonal---perspective, where we do not look at individual
coin tosses but at the global entropy instead. We show how two
distrusting parties can produce a common high-entropy source, where
the entropy is an arbitrarily small fraction below the maximum (except with
negligible probability).

\end{abstract}


\section{Introduction}
\label{sec:intro}
\input{sec-intro}

\subsection{Previous Work}
\label{sec:prevwork}
\input{sec-prevwork}

\section{Preliminaries}
\label{sec:prelims}

\input{sec-prelims}

\section{Sampling from a Quantum Population with a \emph{Mixed}
  Reference State}
\label{sec:samplingstrats}

\input{sec-protocols}



\section{Main Result}
\label{sec:mr}

\input{sec-mainresult}

\section{Two-Party Randomness Generation}
\label{sec:cointoss}
\input{sec-cointoss}
\input{sec-conclusion}

\printbibliography

\appendix

\section{Permutation Invariance of Sampling Protocols}
\label{app:addproofs}
\input{app-proofs}

\section{Additional Proofs}
\label{app:addproofs2}
\input{app-proofs2}

\end{document}

%% file: sec-intro.tex
\subsection{Background and Motivation}

Certifying correctness by means of cut-and-choose techniques is at the
core of many -- classical and quantum -- cryptographic protocols. This
goes back as far as Yao's garbled circuits, introduced in the 80s,
where cut-and-choose is the main technique used to obtain active
security. Even more so, cut-and-choose is at the very heart of
essentially any quantum-cryptographic protocol, where participants are
often asked to prepare states that agree with some specification.
Certifying that quantum states satisfy this specification is essential
to proving the security of these protocols.

Underlying these techniques is one of the most fundamental tasks in
statistics: sampling. It allows one to infer facts about a large set
of data by only looking at a small subset of it. For example, one can
estimate the number of zeros in an $n$-bit string with very high
accuracy by looking only at a small, randomly selected subset of the
bits. This is also true in quantum mechanics: given an $n$-qubit
system, one can infer that it is almost entirely contained in a
subspace $\Span \{ \ket{s} : s \text{ is a bitstring with } (\delta
\pm \epsilon)n \text{ 1's} \}$ by measuring a small subset of the
qubits and observing that a fraction $\delta$ of the bits are
ones~\cite{bouman-fehr}.

One thing that a classical sampling procedure {\em cannot} do, however, is to infer the probability distribution from which the bitstring was generated. While a sampling procedure might be able to tell us that a bitstring contains roughly $n/2$ zeros and $n/2$ ones, that does not mean that it originally came from $n$ fair coin flips\,---\,for all we know, it might be a fixed string that happens to have the right number of zeros and ones. If we were somehow able to do this, it would have interesting consequences for cryptography: for instance, we could get a coin-flipping protocol by getting one party to generate the coin flips, send them to the other party, and have the other party perform this hypothetical sampling procedure to certify that most of the bits received indeed came from fair coin flips.

While this is clearly impossible in the classical case, it turns out that, perhaps surprisingly, this makes sense in the quantum scenario. This is due to the phenomenon of \emph{purification}: given a mixed quantum state $\rho_A$ on system $A$ (which corresponds to a probability distribution on quantum states), it is possible to define a bipartite \emph{pure} (i.e.~deterministic) state $\ket{\psi}_{AR}$ which is in the same mixed state as $\rho_A$ when looking at $A$ only. Hence, one can certify that $A$ is in the mixed state $\rho_A$ by asking someone to produce the purifying system $R$ and measuring that the combined system $AR$ is indeed in state $\ket{\psi}_{AR}$. To give a more concrete example, suppose $\rho_A$ is a uniformly random qubit, i.e.~$\rho_A = \frac{1}{2} \proj{0} + \frac{1}{2} \proj{1}$. Then, the pure state $\ket{\Phi}_{AR} = \frac{1}{\sqrt{2}} (\ket{00} + \ket{11})$ purifies it, and 
checking that $AR$ is in state $\ket{
\Phi}$ certifies that $A$ was uniformly distributed in the first place. Note also that one does not need to trust the party who gives us the purification, making this suitable for an adversarial setting.

This leads to the following natural sampling protocol.  Consider a
sampler Sam who holds an arbitrary quantum state $\rho_{A^n}$ on $n$
subsystems, prepared by a possibly dishonest prover Paul. Sam would
like to certify that this state is close to the ideal mixed state
$\varphi^{\otimes n}$, possibly with errors on a small number of
positions, for a given reference state $\varphi$. To do this, he
selects a small subset of $k$ positions at random, and he asks the
distrusted prover Paul to deliver the purifying systems $R^k$ for these
positions. He then measures the POVM $\{ \proj{\varphi}_{AR}, \ident -
\proj{\varphi}_{AR} \}$ on each of the selected systems in the sample
to ensure that all of them are in the state $\ket{\varphi}_{AR}$ which
purifies $\varphi_A$. He rejects if any errors are detected.

  We emphasize that for verifying a \emph{mixed} reference state,
interaction with a prover is necessary, as there is no local
measurement on Sam's side that can distinguish between the correct
state $\varphi^{\otimes n}$ and a state that consists of the
eigenvectors of $\varphi$ in the correct proportions (i.e., according
to the corresponding eigenvalues).

\subsection{Our Contribution}

In the first part of the paper, we investigate this type of sampling
procedure in detail. Several challenges arise in the analysis of this
protocol. First, defining what we mean when we say that the sampling
works is not trivial. In the case of regular quantum sampling, we
usually want to say that the state has a very small probability of
being outside of a typical subspace that corresponds to the statistics
that we have observed. For mixed states, this definition fails
completely: for instance, in the case of certifying uniformly random
qubits, this typical subspace would actually be the entire space,
yielding a vacuous statement. We might then be tempted to include the
purifying systems in the definition of the typical subspace, but then
we have no guarantee that an adversarial prover will respect the
structure we want to impose on his part of the state---we don't even
know that it consists of $n$ subsystems. A second difficulty comes
from the fact that the prover might not necessarily want to provide
the state that gives him the best chance of passing the test, even if
he has it. If we again look at the case of certifying uniformly random
qubits, even if Sam has the ideal state before the sampling begins,
Paul might want to bias the outcome, for example by passing the test
if he measures $\ket{0}$ on all of the non-sampled qubits, and failing
on purpose otherwise. Because of these difficulties, our main result
does not 
follow 
from traditional sampling theorems.

We overcome these challenges and present a general class of mixed
state certification protocols which contains the natural protocol
described above. We show that any protocol that fits this class, and
that satisfies the simple criteria of being invariant under
permutations and performing well on i.i.d. states, allows us to
control the post-sampling state in a meaningful way. A positive
consequence of this modular analysis is that previous results on
\emph{pure state} certification also fit our framework, and thus fall
under a special case of our analysis -- just as pure states are a
special case of mixed states. Because pure state certification has
already found many applications in
cryptography~\cite{bouman-fehr,damgaard2015orthogonal,dfls16,fkszz13,winkler2014efficiency},
the fact that we recover it as a special case positions our result as
a powerfool tool for quantum cryptography.

The second part of the paper is devoted to applying this result to 
coin flipping---or \emph{randomness generation}.
Given that strong coin tossing is known to be
impossible, it is natural to ask ``how close can we get?". This
question has been well studied and is nowadays well understood from
the perspective of the bias of individual coin tosses (see
Section~\ref{sec:prevwork} below).  We approach and answer this
question from a different---and somewhat orthogonal---perspective,
where we do not optimize individual coin tosses but the global entropy
instead. From this entropic perspective, we show that ``the next
best'' after strong coin tossing is possible.  We show that the
coin-flipping protocol loosely described above allows two distrusting
parties to produce a common high-entropy source, where the entropy is
an arbitrarily small fraction below the maximum (except with negligible
probability).

Our protocol for the task of two party randomness generation
outperforms any classical protocol in the information theoretical
setting. The trivial classical protocol---where each party tosses
$n/2$ unbiased coins and the output is the result of the $n$ tosses---is
optimal for this task~\cite{hofheinz2006possibility}.

The paper is organized as follows. First, in the next
subsections, we discuss some previous work in the area and the
relevance of our work for cryptography. In Section~\ref{sec:prelims},
we introduce the notation and recall some useful
facts. Section~\ref{sec:samplingstrats} presents the main result in
more detail.  The proof of our main result follows in
Section~\ref{sec:mr}. The coin-flipping protocol described above is
then presented in Section~\ref{sec:cointoss}.


%% file: sec-prevwork.tex
Classical sampling results have been around since the foundations of modern probability theory, dating back to the work of Bernstein, Hoeffding and Chernoff on concentration of measure in the 1920s and 1930s. More recently, several quantum generalizations of these classics have been proven. These generalizations include, for instance, Ahlswede and Winter's operator Chernoff bound~\cite{ahlswede-winter} and the quantum Chernoff bound of~\cite{quantum-chernoff}. However, these generalizations are not easily amenable to giving results about sampling, unlike their classical counterparts. Other quantum results can be used to analyze sampling in certain contexts, such as quantum de Finetti theorems for quantum key distribution~\cite{renner-phd,r07,rkc08}.

But perhaps the most direct analogues of the classical sampling results are those of~\cite{bouman-fehr}. There, the authors give a generic way to transpose classical sampling procedures to the quantum case. Roughly speaking, they show that if a classical sampling protocol says that a string of random variables $X_1,\cdots,X_n$ is contained in some ``good'' subset $\mathcal{X}_{\textrm{good}}$ except with negligible probability, then the quantum version of the same sampling procedure (defined in a precise way in \cite{bouman-fehr}) would say that the final state $\rho_{X_1,\ldots,X_n}$ is almost entirely contained in the good subspace $\Span\{ \ket{x_1} \otimes \cdots \otimes \ket{x_n} : x_1,\cdots,x_n \in \mathcal{X}_{\textrm{good}}\}$, except with negligible probability. This ``good'' set would normally correspond to strings that are consistent with what was observed in the sample. Our main result can be viewed as extending this to the case of mixed state sampling.

Our main application, coin flipping, also has a long history. The
basic task was first defined in 1981 by Manuel Blum~\cite{b81}. Since
the early 2000's, it has received a lot of attention in the quantum
cryptography community, as it is one of the most natural tasks for
which quantum protocols can perform something that is impossible
classically. There are two versions of coin flipping: \emph{strong}
coin flipping, in which we require the protocol to be equivalent to a
black box that produces the coin flip and distributes the result, and
\emph{weak} coin flipping, in which each participant has a known
preferred outcome and must be prevented from biasing the outcome in
that direction. Several quantum protocols for strong coin flipping
have been developed with various biases~\cite{sr01,a02}, but a
fundamental lower bound of $(\frac{1}{\sqrt{2}}-\frac 12)$ on the bias
of such protocols was proven in~\cite{k03-2} (see
also~\cite{gw07}). Finally, a protocol with a bias matching the lower
bound was proven in~\cite{ck09}. For weak coin flipping, we have had
several protocols~\cite{kn04,sr02,m04,m05}, again with various biases,
but this time culminating in a protocol with arbitrarily small
bias~\cite{m07}. Quantum coin flipping has even been implemented in
the lab~\cite{pjlcltkd14}. Here, we go in a somewhat different
direction: we show that even though strong coin flipping with
negligible bias is impossible without assumptions, two distrustful
parties can produce a common string of min-entropy arbitrarily close
to maximum.

  A strong quantum coin tossing protocol
  using ideas similar to that of the 
protocol described in Section~\ref{sec:cointoss} has been previously
considered by H\o{}yer and Salvail (unpublished) for achieving
in a slightly simpler way the same $\frac{1}{4}$ bias than the one in~\cite{a02}. 
Alice prepares two EPR pairs and sends one half of each to Bob. Bob picks at
random one qubit out of the two and verifies that Alice holds the
corresponding purification register of an EPR pair by asking her to
measure it in a random BB84 basis before comparing the result with his
own. If this test succeeds, Bob gets some evidence that the
remaining pair of qubits can be used as a coin toss after measuring it
in the canonical basis. Our protocol extends this test to a random
sample of a population of $N$ qubits, increasing the confidence that
Bob has about the remaining qubits being ``close'' to ideal coin
tosses when the test is successful.

\subsection{Applications to Cryptography}
\label{sec:applications}

\input{sec-applications}


%% file: sec-applications.tex
\paragraph{Sampling with a Pure Reference State. }
\label{sec:purestates}

Previous results on sampling from a quantum population have dealt with
\emph{pure} reference states. In this case, the sampler can choose its
sample and perform local measurements on the sampled positions without
any help from the prover.  This setting allows for standard classical
tools such as Hoeffding's inequality to be used to derive the
probability that the sampled positions' proximity to the reference
state is not a good indicator for the unsampled positions' proximity
to the same reference state.

Since pure states are a special case of mixed states, a natural
property that we would want for our mixed state sampling result is to
recover a statement similar to the one for pure state sampling in the
framework of~\cite{bouman-fehr}.  This is indeed the case when we
restrict our attention to the task of certification, i.e. when we do
not tolerate any error in the sample. Although our results do not use
the same tools, and are expressed in terms of a
\emph{post-selected} operator instead of in terms of proximity to an
ideal state (see Sect.~\ref{sec:samplingstrats}), we recover a
statement equivalent to that of~\cite{bouman-fehr}, albeit with
slightly worse parameters, when we apply our results to pure reference
states.  Since most
applications~\cite{bouman-fehr,dfls16,fkszz13,winkler2014efficiency}
of pure state sampling has been in the setting of certification, our
results can also be used to prove those applications.

\paragraph{Sampling with a Distributed Pure Reference State. }

Our mixed state sampling result is also applicable to an instance of
pure state certification that falls outside  the framework
of~\cite{bouman-fehr} and which was presented and analyzed in an ad
hoc way in~\cite{damgaard2015orthogonal}. Their sampling algorithm was
used as part of a protocol for leakage resilient computation.

The sampling task considered in~\cite{damgaard2015orthogonal} is as
follows: spacially separated Alice and Bob want to certify that their
joint registers -- which was prepared by an untrusted third party --
is of the form $\ket\varphi_{AB}^{\otimes n}$ for some entangled state
$\ket \varphi$ where Alice holds the $A$ part of each of the $n$
states and Bob the $B$ part. The fact that the state is distributed
between Alice and Bob means that the techniques of~\cite{bouman-fehr}
do not apply: the two samplers cannot perform a projective measurement
to check that their shared registers are in the reference state
$\ket\varphi_{AB}$.

Our results of Sect.~\ref{sec:mr} only requires that the sampling
protocol's verification procedures is invariant under the permutation
of the quantum population, and that it aborts when performed on an
\emph{obviously bad} state. 
Since the pure state certification protocol
of~\cite{damgaard2015orthogonal} satisfies these properties, our
techniques readily apply and can be used to analyze their protocol.

\paragraph{Application to Two-Party Computation. }

In \cite{salvail2015quantifying}, the power of quantum communication
for secure unconditional two-party computation is investigated.  Among
other results,
it was shown that \emph{correct} quantum implementations of two-party
classical cryptographic
primitives 
must leak at least some minimal amount of information to one of the
parties. 
For example, randomized variants\footnote{Variants where the
  primitives considered are applied to random inputs.} of
one-out-of-two OT and secure AND sharing must leak at least
$\frac{1}{2}$ bit on average.
Protocols exist in the quantum honest-but-curious model that minimize
the amount of leakage for a given primitive. The simplest such
protocol consists of an adversary preparing and distributing an
\emph{embedding} of the primitive. An embedding of a cryptographic primitive
is a pure state that yields the correct outcomes when measured in the
computational basis, i.e. from each party's point of view, the state
shared before the final measurement is a purification of the
probability distribution for this party's output.

A protocol that achieves
minimal leakage against \emph{active} adversaries under the sole
assumption that the parties have access to strong strong
coin-tosses is easily obtained from mixed-state certification. 
One of the parties would generate many copies of the
embedding of the primitive that minimizes leakage and the other party
certifies correctness using our sampling procedure. They then choose
one of the remaining embeddings, the target embedding, and measure it; the outcome acts as
the output of the protocol. If the sampling succeeds, the unsampled
positions are close to ideal embeddings from the sampler's
perspective and randomly picking the target embedding 
would then have close to minimal leakage with
good probability. However, without additional resources, an
adversary (the sampler say) could measure its part of a few embeddings before choosing 
the target embedding as one  
that produces the output the 
adversary wants to see.
Coin-tosses are therefore required to pick the target embedding without  bias.


%% file: sec-prelims.tex
\subsection{Notation}
\label{sec:notation}

Let $\hilbert_A,\hilbert_B$ be two Hilbert spaces, we write
$L(\hilbert_A, \hilbert_B)$ for the set of linear operators from
$\hilbert_A$ to $\hilbert_B$ and we write $L(\hilbert_A)$ for
$L(\hilbert_A, \hilbert_A)$.  Let $\leqdensity\hilbert$ be the set of
positive semi-definite operators with trace less than or equal to 1,
and let $\density\hilbert$ be the set of density operators on
$\hilbert$.  The set of isometries from $\hilbert_A$ to $\hilbert_B$
is denoted $U(\hilbert_A, \hilbert_B)$. We use the notation
$U_{A\rightarrow B}$ to illustrate that $U_{A\rightarrow B}\in
U(\hilbert_A, \hilbert_B)$. When there is no ambiguity from doing so,
we write $U_A$ instead of $U_{A\rightarrow B}$. For an arbitrary
isometry $U$, we sometimes write $[U](\rho)$ as shorthand for $U\rho
U^\dagger$.  For a pure state $\ket\psi$, we write $\psi$ as shorthand
for $\proj\psi$ when this creates no ambiguity. For a linear operator
$A$, $\|A\|_1:= \trace{\sqrt{A^\dagger A}}$ denotes the \emph{trace
  norm}. We denote $\id_A$ as the identity operator on $\hilbert_A$
and $\idcptp_A$ as the CPTP map that acts trivially on register $A$.

We let $[n]:=\{1,\dots,n\}$ denote the set of the first $n$ positive
integers for $n\in \naturals$. For a fixed finite set $Y$ and any subset
$X\subseteq Y$, $\bar X$ denotes the complement of $X$ in $Y$,
i.e. $\bar X= Y\setminus X$. Let $h(p):= -p\log_2(p)
-(1-p)\log_2(1-p)$ be the binary entropy function; we make use of the
fact that $\binom n{\beta n}\leq 2^{h(\beta)n}$ for $0<\beta<1$.

Let $A$ be a quantum register, we use the notation $A^n$ to denote $n$
identical copies of $A$ and label them $A_1,\dots,A_n$ when the need
arises to distinguish individual registers. For $t\subseteq [n]$, we
write $A_t$ as the composite register containing registers $A_i$ for
each $i\in t$.

\subsection{Permutation Invariance and the Symmetric Subspace}
\label{ssec:symdef}

Let $\symgrp n$ denote the symmetric group on $n$ elements and let
$A_1,\dots, A_n$ be $n$ quantum registers with identical state space
$\hilbert$. For $\pi\in \symgrp n$, we use the same symbol to denote the
unitary operation that acts on $\hilbert^{\otimes n}$ by 
\begin{equation}
  \pi (\ket{\phi_1}_{A_1}\otimes\dots\otimes \ket{\phi_n}_{A_n}) =
  \ket{\phi_{\pi^{-1}(1)}}_{A_1} \otimes \dots \otimes
  \ket{\phi_{\pi^{-1}(n)}}_{A_n}
  \enspace .
\label{eq:permutationaction}
\end{equation}

\begin{definition} 
  A density operator $\rho\in \density{\hilbert^{\otimes n}}$ is called
  \emph{permutation invariant} if $\pi \rho\pi^\dagger=\rho$ for all
  $\pi\in \symgrp n$.

  The \emph{Symmetric subspace} of $\hilbert^{\otimes n}$, denoted
  $\sym^n(\hilbert)$, is the space spanned by all permutation
  invariant vectors of $\hilbert^{\otimes n}$, i.e. all vectors $\ket
  \phi\in\hilbert^{\otimes n}$ such that $\pi\ket\phi=\ket\phi$ for
  any $\pi\in \symgrp n$.
\end{definition}

Although not all permutation invariant operators have support in the
symmetic subspace, the next lemma asserts that they have a
purification that does.
\begin{remark}[\cite{renner-phd,ckmr07}]\label{lem:purif}
  For any permutation invariant density operator $\rho_{A^n}$ on
  $\hilbert_A^{\otimes n}$ there exists a pure state $\ket{
    \rho_{A^nB^n}}\in \sym^{n}(\hilbert_A\otimes\hilbert_B)$ where
  $\hilbert_A\simeq \hilbert_B$, such that $\trace[B^{n}]{
    \rho_{A^nB^n}}=\rho_{A^n}$.
\end{remark}

\begin{remark}[\cite{renner2010,renner-phd}]\label{rem:projsymsubspace}
  Let $\hilbert$ be a $d$-dimensional Hilbert space. The projector
  onto the symmetric subspace $\sym^n(\hilbert)$ can be expressed
  as $$ c_{n,d} \int \proj\theta^{\otimes n} d\ket\theta$$ where
  $d\ket\theta$ is the measure on the set of pure states of $\hilbert$
  induced by the Haar measure on the set of unitaries acting on
  $\hilbert$ 
  and where $c_{n,d} := \binom{n+d-1}{n}\leq (n+1)^{d-1}$ is the
  dimension of $\sym^n(\hilbert)$.
\end{remark}

\subsection{Mathematical Tools and Definitions}
\label{sec:tools}

We say that an operator $\tilde \rho_B$ is \emph{post-selected} from
register $A$ of $\rho_{AB}$ if there exists a POVM element $0\leq
E_A\leq \id_A$ such that $\tilde\rho_B= \trace[A]{(E_A \otimes
  \id_B)\rho_{AB}}$.  The following remark on relation between the
reduced operator of a joint system before and after a post-selected
measurement takes place will be useful throughout this paper.
 \begin{remark}\label{rem:fq09h4g} Let $\rho_{AB}$ be an arbitrary positive
   semi-definite operator on registers $AB$. Let $0\leq E_A\leq \id_A$
   be a positive semidefinite operator acting on register $A$. Then it
   holds that
   \begin{equation*}
     \trace[A]{(E_A\otimes \id_B) \rho_{AB}} \leq
     \trace[A]{\rho_{AB}}\enspace. 
   \end{equation*}
 \end{remark}

 The following observation shows that there is a strong relation
 between post-selected operators and upper-bounded operators.
\begin{proposition}\label{prop:leqequivpostsel}
  Let $c\geq 0$ and let $\rho_Q, \sigma_Q$ be two positive
  semi-definite operators. Then $\rho_Q\leq c\cdot \sigma_Q$ if and
  only if for any purification $\ket{\sigma_{R_1Q}}$ of $\sigma_Q$ and
  $\ket{\rho_{R_2Q}}$ of $\rho_Q$, there exists a linear operator
  $A_{R_1\rightarrow R_2}$ such that $A_{R_1}^\dagger A_{R_1}\leq
  \id_{R_1}$ and
  \begin{equation}\label{eq:post-selection}
    \ket{\rho_{R_2Q}}= \sqrt{c}\cdot (A_{R_1\rightarrow R_2}\otimes \id_Q) \ket{\sigma_{R_1Q}}
  \end{equation}
\end{proposition}

The following Proposition is a generalization of a Lemma that appeared
in~\cite{bouman-fehr}, which is itself has roots
in~\cite{renner-phd}. A direct consequence of this Proposition is that
a superposition of a few states can be \emph{approximated} by a
mixture of the same few states.
\begin{proposition}
  \label{prop:smallnumbterms}
  Let $\{\ket {\psi_i}\}_{i\in \mathcal J}$ be a family of vectors
  living on a Hilbert space $\hilbert$ indexed by some finite set
  $\mathcal J$. Define operators $$\rho = \sum_{i,j\in \mathcal J}
   \ketbra{\psi_i}{\psi_j} \text{ and }
  \rho^{mix}= \sum_{i\in \mathcal J}
  \proj{\psi_i}\enspace .$$ Then, $\rho \leq |\mathcal J|
  \cdot\rho^{mix}$.
\end{proposition}

\begin{definition}[Quantum ``Hamming Ball'']\label{def:quanthammingball}
  Let $\ket \Psi\in \hilbert^{\otimes n}$ for $n\in \naturals$ and let
  $r\in [n]$. We define the quantum \emph{Hamming ball} of radius $r$
  around $\ket \Psi$, denoted $\Delta_r(\ket\Psi)$, as the space
  spanned by all vectors of the form $U\ket \Psi$ where $U$ is a
  unitary that acts as the identity on at least $n-r$ subsystems.  

  For the special case where $\ket\Psi=\ket \nu^{\otimes n}$,
  \begin{equation*}
    \Delta_r(\ket\nu^{\otimes n}) = \spn\{ \pi( \ket \nu^{\otimes
      n-r}\otimes \ket u)\;:\; \ket u\in \mathcal{B}, \pi \in \symgrp n\}
  \end{equation*}
  where $\mathcal{B}$ is an orthonormal basis of $\hilbert^{\otimes
    r}$.
\end{definition}

The projector onto the quantum Hamming ball of radius $r$ around an
i.i.d.  state $\ket\nu^{\otimes n}\in \hilbert_{A_1}\otimes\dots
\otimes \hilbert_{A_n}$ can be written as
\begin{equation*}
  \mathbb P_{A^n}^{r, \ket\nu} = \sum_{E\subseteq [n]\;:\; |E|\leq r}
  \left(\bigotimes_{i\in E} (\id-\proj \nu)_{A_i} \bigotimes_{i\notin E} \proj
    \nu_{A_i}\right)\enspace .
\end{equation*}

The following Lemma says that $n$ i.i.d. copies of a state close to
$\ket \nu$ is almost entirely contained in a Hamming ball around $\ket
\nu^{\otimes n}$.
\begin{lemma}
  \label{lem:puriftauisideal}
  Let $\ket{\nu}, \ket \theta\in \hilbert$ be such that
  $|\braket{\theta}{\nu}|^2\geq 1-\epsilon$.
  Then, for any $\alpha>0$, 
  \begin{equation*}
    \trace{\mathbb P^{r, \ket\nu}\cdot
       \proj\theta^{\otimes n}}\geq 1-\exp(-2\alpha^2 n)
  \end{equation*}
  where $\mathbb P^{r, \ket\nu}$ is the projector onto
  $\Delta_{r}(\ket \nu^{\otimes n})$ for $r=
  (\epsilon+\alpha) n$.
\end{lemma}


%% file: sec-protocols.tex
The task we analyze can be understood as an interactive game between two
participants: a \emph{prover} Paul, and a \emph{sampler} Sam.
Paul is supposed to prepare multiple copies of some
{\em reference state} $\varphi$ before sending them to Sam, 
and the purpose of the game is for Sam to detect when the state produced by Paul is (close to) what it is supposed to be, no matter how maliciously Paul behaves. 
Here, the reference state $\varphi$ may be an arbitrary but known {\em mixed}
state. 
A canonical example of such a quantum sampling protocol  is depicted in
Fig.~\ref{fig:samplingprot}. It consists of
Sam asking Paul to deliver the purification registers 
of $k$ randomly chosen positions. Sam then measures these purifications 
in order to learn if they were in the right state.%
\footnote{Note that there is no loss in generality in announcing the positions that Sam wants to check {\em in one go} as is done in Fig.~\ref{fig:samplingprot}, compared to announcing them {\em one-by-one}; doing it the latter way only makes it harder for Paul. }

\begin{figure}[h]
  \begin{framed}
  \begin{center}
    \purifbased
  \end{center}
    \begin{enumerate}
    \item\label{step:StateDistr} Paul prepares $N$ copies of the
      purification $\ket{\varphi_{P S}}$ of $\varphi_{ S}$, he sends
      $N$ registers in state $\varphi_S$ labeled $S_1$ to $S_N$ to Sam
      and keeps the corresponding purification registers $P_1$ to
      $P_N$.
    \item \label{st:f048ht3} Sam picks a subset $t\subseteq [N]$ of
      size $k$ uniformly at random.
    \item Sam sends $t$ to Paul and asks him to send him the
      purification registers $P_i$ for $i\in t$. \label{purifsend} 
    \item\label{step:test} Sam measures each register $P_iS_i$ for $i\in
      t$ using projective measurement $\{\proj\varphi_{ P S},
      \id_{ P S}-\proj\varphi_{ P S}\}$. 
      Sam accepts if he observed $\proj\varphi^{\otimes k}$,
      otherwise, he rejects.\label{st:f0q4gh}
    \end{enumerate}
  \end{framed}
  
  \caption{The purification-based mixed state quantum sampling
    protocol with reference state $\varphi_{ S}$. Paul and Sam
    need to have previously agreed on a purification $\ket{\varphi_{ P
        S}}$ of $\varphi_{ S}$.}
  \label{fig:samplingprot}
\end{figure}

In the extreme case of a reference state that is empty on Paul's side,
and thus pure on Sam's side (and so there is no purification for Paul
to provide in step~\ref{st:f048ht3}), the sampling protocol of
Fig.~\ref{fig:samplingprot} pretty much coincides with the pure-state
sampling procedure considered and analyzed in~\cite{bouman-fehr}.
For a true mixed reference state, however, it is significantly
harder to prove that the sampling protocol ``does its job'' 
because of
the additional freedom that Paul has in preparing the purification
registers that may depend on the choice of $t$.  This very much renders the
techniques from~\cite{bouman-fehr} useless. Indeed, the idea of the
analysis in~\cite{bouman-fehr} was to assume, for the sake of the
argument, that the positions outside of $t$ are measured as well, and
then to delay the choice of $t$ to after the measurement so as to
reduce to a classical sampling procedure. Because of
Paul's freedom in choosing the purifications dependent on $t$, it
makes no sense to speak about the outcome of the reference measurement
$\{\proj\varphi, \id-\proj\varphi\}$ {\em before} $t$ is chosen, or
about the measurement being applied to a position {\em outside}
of~$t$.  As such, we need an entirely different approach.

Before worrying about analyzing the mixed-state sampling protocol of
Fig.~\ref{fig:samplingprot}, we first need to specify what it should
actually mean for it to ``do its job''; this is not entirely obvious.
Intuitively, we want that after the sampling, if Sam accepts
then his part of the state should be ``somehow close'' to what it is
supposed to be, namely $\varphi^{\otimes n}$ where we set $n = N-k$.
However, Paul can obviously cheat in a small number of positions,
i.e., start off with a state that consists of i.i.d. copies of
$\ket{\varphi}$ except for a small number of positions where the state
may deviate arbitrarily, and he still has a fair chance of not being
caught. Of course, the same holds for a mixture of such states, and
therefore, by purification, also for a superposition of such states.
This motivates the definition below of an ``ideal state'', which
captures the best we can hope for. The formal statement of what the
sampling protocol of Fig.~\ref{fig:samplingprot} achieves is then in
terms of controlling Sam's part of the state after the protocol by
means of Sam's part of such an ideal state. This is somewhat similar
in spirit as the approach in\cite{bouman-fehr} for pure-state
sampling, though there are some technical differences.

\begin{definition}[Ideal States]\label{idealdef}
 For $\epsilon>0$, 
  a state $\psi_{S^n}\in \mathcal{D}_{\leq}(\hilbert_{S}^{\otimes n})$ is said to be \emph{$\epsilon$--ideal} if there
  exists a  purification $\ket{\psi_{RP^nS^n}}$ of $\psi_{S^n}$
   such that
    \begin{equation*}
      \ket{\psi}_{RP^nS^n}\in \hilbert_R\otimes
        \Delta_{\epsilon n}(\ket\varphi^{\otimes n}_{P^nS^n} )\enspace.
      \end{equation*}
      We loosely say that  $\psi_{S^n}$ is \emph{ideal} when it is 
      $\epsilon$--ideal for \emph{small}  $\epsilon$.
\end{definition}

Our analysis of the sampling protocol described in
Fig.~\ref{fig:samplingprot} (and some variants of it) preserves many
aspects of the operational interpretation provided in
\cite{bouman-fehr} when sampling with respect to a pure reference
state.  We establish that Sam's \emph{subnormalized} final state of
register $S^n$ upon acceptance is overwhelmingly close to an ideal
state. The subnormalized state is simply the state Sam is left with
when he accepts scaled down by the probability of acceptance (i.e. its
trace corresponds to the probability for Sam to accept).
Let $d:=\dim{(\hilbert_{S})}$ be the size of the register holding  $\varphi_S$ and
let  $\epsilon>0$ be a parameter. 
Informally, our main theorem 
(Theorem~\ref{thm:mainresultunpermuted} and Corollary~\ref{thm:mrupperbound})
establishes that Sam's subnormalized final state upon acceptance 
$\rho_{S^n}^\acc\in \mathcal{D}_{\leq}(\hilbert_{S}^{\otimes n})$
is such that
\begin{equation}\label{intuitmain}
\rho_{S^n}^\acc \leq (N+1)^{d^2-1}\psi_{S^n}+\sigma_{S^n}\enspace,
\end{equation}
where ${\psi}_{S^n}$ is ideal 
and $\|\sigma_{S^n}\|_1$ is negligible in $N$.

Any state $\rho_{S^n}^\acc$ that satisfies (\ref{intuitmain}) can be
considered to be an ideal state in many applications.  Let
$\mathcal{Q}$ be a completely positive trace non-increasing
super-operator modelling a task that we would like to apply upon
$\rho_{S^n}^\acc$. Suppose that $\mathcal{Q}$ behaves nicely when it
is executed from an ideal state $\psi_{S^n}$. That is, the bad event
represented by a POVM element $E_{\text{bad}}$ has negligible
probability on the ideal state $p^{\text{id}}_{\text{bad}} :=
\trace{E_{\text{bad}} \mathcal{Q}(\psi_{S^n})} \leq 2^{-\alpha N}$ for
$\alpha>0$.  Running $\mathcal{Q}$ upon $\rho_{S^n}^\acc$ instead
produces the state $\mathcal{Q}(\rho_{S^n}^\acc) \leq
\mathcal{Q}((N+1)^{d^2-1}\psi_{S^n} + \sigma_{S^n})$.  We then have
that the probability of the bad event in the real case is
$p^{\text{real}}_{\text{bad}} := \trace{E_{\text{bad}}
  \mathcal{Q}(\rho_{S^n}^\acc)} \leq
(N+1)^{d^2-1}p^{\text{id}}_{\text{bad}}+ \|\sigma_{S^n}\|_1$, which
remains negligible when $p^{\text{id}}_{\text{bad}}$ is negligible and
$d$ is small enough (i.e. a constant).  In other words, any negligible
upper bound on the probability of some ``bad'' event occurring when
processing the ideal state translates to a negligible upper bound on
the ``bad'' event when processing the real state instead.  In these
cases, it is good enough to analyze the ideal state, for which an
analysis is typically simpler because of the specific form of the
state as given by Definition~\ref{idealdef}.

Our main result can also be interpreted 
as a statement about Paul and Sam's joint state 
when Sam accepts. To do so, we invoke Proposition~\ref{prop:leqequivpostsel}
upon (\ref{intuitmain}). For the sake of simplicity, assume that 
$\rho_{S^n}^\acc \leq c\cdot \psi_{S^n}$, which is essentially
what (\ref{intuitmain}) means
for $c :=(N+1)^{d^2-1}$. Proposition~\ref{prop:leqequivpostsel} 
then establishes the existence of  a linear operator
$A$ acting upon registers $R P^n$ for which $A^\dagger A\leq \id$
such that
\begin{equation}\label{yetanotherinterpretation}
 \ket{\rho^\acc}_{R P^n S^n}  = \sqrt{c} (A\otimes \id_{S^n})\ket{\psi}_{R P^n S^n} \enspace,
\end{equation}
where $\ket{\rho^\acc}_{R P^n S^n}$ and 
$\ket{\psi}_{R P^n S^n}$ are purifications of
$\rho^\acc_{S^n}$ and $\psi_{S^n}$, respectively.
The operator $E:=AA^{\dagger}$ can be viewed as 
the outcome of a POVM applied upon registers
$R P^n$ implemented by the detection operator $A$. 
It  follows from (\ref{yetanotherinterpretation})
that $\rho^\acc_{R P^n S^n}$ can be obtained with a non-negligible probability of success $1/c$ by applying a measurement upon an ideal state $\psi_{R P^n S^n}$.
Therefore, any application having a negligible probability for Paul
to generate a \emph{bad} shared state from an ideal one has also a negligible probability 
to generate a \emph{bad} shared state from the real one.

\subsection{Sampling Protocol Using Local Measurements and Classical
  Communication}
\label{sec:EPRsampling}

Our analysis of mixed state sampling protocols is not limited to the
protocol of Fig.~\ref{fig:samplingprot}. In Sect.~\ref{sec:mr}, we
show that any sampling protocol that satisfy certain criteria can be
analyzed using our techniques. One such protocol is the one depicted
in Fig.~\ref{fig:EPRsampling}. It is a protocol for certifying that
Paul prepares---and purifies---halves of EPR pairs that requires
only local operations and classical communication (LOCC) after the
initial state preparation and distribution phase. EPR pairs are states
of the form $\ket{\Phi^+}=\frac 1{\sqrt 2}(\ket {00}+\ket {11})$ that
have the unique property that measurements in both the computational
and diagonal bases are perfectly correlated. The protocol exploits
this fact in the following way: for each position in the sample, Sam
asks Paul for the result of measuring his purifying register in a
random basis, and checks that this result corresponds to his own
measurement in the same basis.

\begin{figure}[h]
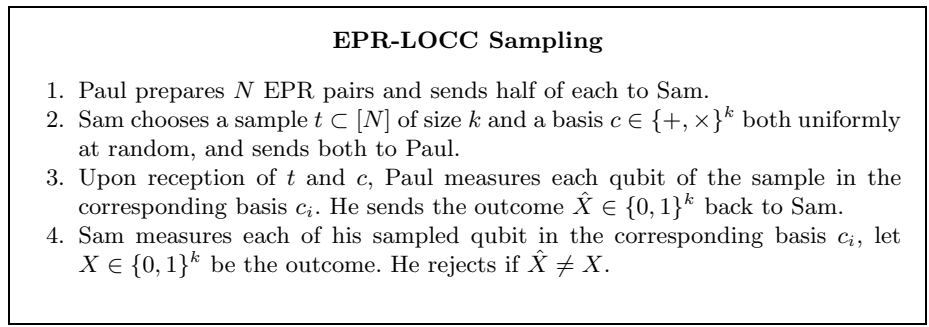

  \centering
  \begin{framed}
    \eprlocc
    \begin{enumerate}
    \item Paul prepares $N$ EPR pairs and sends half of each to Sam. 
    \item Sam chooses a sample $t\subset[N]$ of size
      $k$ and a basis $c\in\{+,\times \}^k$ both uniformly at random,
      and sends both to Paul.
    \item  Upon reception of $t$ and $c$, Paul
      measures each qubit of the sample in the corresponding basis
      $c_i$. He sends the outcome $\hat X\in \bool^k$ back to Sam.
    \item  Sam measures each of his sampled qubit
      in the corresponding basis $c_i$, let $X\in \bool^k$ be the
      outcome. He rejects if $\hat X\neq X$.
    \end{enumerate}
  \end{framed}
  \caption{The sampling protocol with local measurements for sampling
    halves of EPR pairs, i.e. with reference state $\varphi=\frac \id
    2$.}
  \label{fig:EPRsampling}
\end{figure}


%% file: sec-mainresult.tex
In this section, we present the techniques that allow to analyze
sampling protocols similar to that of Fig.~\ref{fig:samplingprot}. The
key property of the sampling protocol that makes the tools of this
section applicable is that it is invariant under the permutation of
the sampler's register, up to an adjustment of the adversary's attack
and of the output state.  In order to make this more explicit, we
actually consider and analyze a general class of sampling protocols
that are permutation invariant and perform well on i.i.d. states, and
we then show (1) that the protocol of Fig.~\ref{fig:samplingprot}
falls into that class and (2) that any protocol from that class allows
us to control the post-sampling state the way we want. As an
additional bonus of this modular analysis is that we can then easily
extend our results to other sampling protocols. For instance, the
sampling protocol of Fig.~\ref{fig:EPRsampling} for certifying EPR
pairs presented in Sect.~\ref{sec:EPRsampling} also falls into the
class of protocols that we consider. In that protocol, Paul is not
asked to provide his respective parts of the EPR pairs from within the
sampled subset, but he is instead asked to provide the {\em
  measurement outcome} of those, when measured in a random basis
chosen and announced by Sam, and Sam compares with the corresponding
measurement outcomes on his side.

\subsection{Mixed State Sampling Protocols and Permutation Invariance}
\label{sec:perminvsampling}

The general form of the sampling protocols we consider is depicted in
Fig.~\ref{fig:generalsampling}. For simplicity, we assume that the
protocol always outputs the same number of qudits $n=N-k$, i.e. that
it lives in the Hilbert space $\hilbert_S^{\otimes n}$.  Note that
this means that there is no freedom in the way we choose the sample
$t$; the only permutation invariant probability distribution on the
subsets of $[N]$ of size $k$ is the uniform distribution. We also
assume that $k$ is of the order of $N$.

\begin{figure}[h]
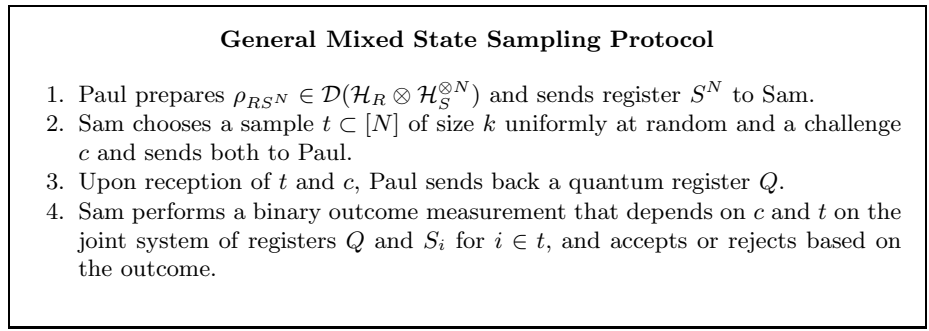

  \begin{framed}
    \begin{center}
      \generalsampling
    \end{center}
    \begin{enumerate}
    \item Paul prepares $\rho_{RS^N}\in \density{\hilbert_R\otimes
        \hilbert_S^{\otimes N}}$ and sends register $S^N$ to Sam.
    \item \label{st:pickc} Sam chooses a sample $t\subset[N]$ of size
      $k$ uniformly at random and a challenge $c$ and sends both to
      Paul.
    \item \label{st:sendpprime} Upon reception of $t$ and $c$, Paul
      sends back a quantum register $Q$.
    \item \label{st:measurec} Sam performs a binary outcome
      measurement that depends on $c$ and $t$ on the joint system of
      registers $Q$ and $S_i$ for $i\in t$, and accepts or rejects
      based on the outcome.
    \end{enumerate}
  \end{framed}
  \caption{The general form of a mixed state sampling protocol for
    sampling a mixed reference state $\varphi$.}
  \label{fig:generalsampling}
\end{figure}

The obvious example instantiation of such a sampling protocol is the
sampling protocol of Fig.~\ref{fig:samplingprot}, where $c$ is empty
and Sam's measurement
consists of projecting onto $\proj\varphi^{\otimes k}$. 
Another example is the one we discuss in Section~\ref{sec:EPRsampling}
for certifying EPR pairs, where
$c$ then is a randomly chosen sequence of bases that specifies how Paul is supposed to measure his parts of the EPR pairs. 

Clearly, for a given instantiation of the general protocol of
Fig.~\ref{fig:generalsampling}, the adversary's attack strategy consists
of the choice of $\rho_{RS^N}$ and of the quantum operation (that
depends on $t$ and $c$) that produces $Q$ in
step~\ref{st:sendpprime}. 

We now define the notion of permutation invariance that sampling
strategies must satisfy for our techniques to apply.

\begin{definition}[Permutation Invariance for Sampling Protocols]
  \label{def:perminv}
  A sampling protocol that implements the framework of
  Fig.~\ref{fig:generalsampling} is \emph{invariant under the
    permutation of the sampler's register} if for any adversarial
  strategy for Paul, the completely positive trace non-increasing map
  $\Eacc$, which represents the output
  state of the sampler when he accepts, 
  satisfies
  \begin{enumerate}
  \item for any input $\rho_{RS^N}\in \density{\hilbert_R\otimes
        \hilbert_S^{\otimes N}}$ there exists $\barEacc$ such that
    \begin{equation}
      \label{eq:0q4gh0awfe}
      \frac 1{n!}\sum_{\pi\in \symgrp n} \proj\pi_\Pi\otimes
      \pi_{S^n} \Eacc(\rho_{RS^N})
      \pi_{S^n}^\dagger=
      \barEacc(\bar \rho_{P^NS^N})
    \end{equation}
    for some symmetric purification $\ket{\bar \rho_{P^N
        S^N}}\in\sym^N(\hilbert_P\otimes \hilbert_S)$ of $ \frac
    1{N!}\sum_{\pi\in \symgrp N} \linebreak \pi_{S^N} \rho_{S^N}
    \pi_{S^N}^\dagger$,
  \item for any $\epsilon>0$, $\|\barEacc(\proj\theta^{\otimes N})\|_1 \leq
    \exp(-\Omega(N))$ whenever $F(\theta_{S}, \varphi_{
      S})^2< 1-\epsilon$, and
  \item $\barEacc$ acts trivially on the unsampled systems, up to
    reordering. Formally, $\barEacc$ satisfies
    \begin{equation*}
      \label{eq:outputnontouche}
      \trace[\Pi]{\barEacc(\proj\theta^{\otimes N}_{PS})} \leq
      \theta^{\otimes n}_S \enspace .
    \end{equation*}

  \end{enumerate}
\end{definition}
The first
criterion effectively requires that 
any attack against the sampling protocol of
Fig.~\ref{fig:generalsampling} 
can be transformed into an \emph{equivalent} attack on a symmetric
state\,---\,up to a random reordering of the positions.
The second criterion demands that Bob rejects with overwhelming
probability in case of an ``obviously bad'' state, i.e., in case of
i.i.d. copies of a state that is far from the reference
state~$\varphi$.  The third criterion simply asks that the sampling
protocol (and the corresponding symmetrized map $\barEacc$) does not measure registers
outside the sample.

From a technical perspective, the first criterion allows us to apply
the observations from Section~\ref{ssec:symdef} to the promised
symmetric state, so that we can upper bound the latter by a convex
linear combination of i.i.d. states, and the second criterion then
allows us to control the ``bad part'' of this convex linear
combination (see Section~\ref{sec:symmetricproof}). What then still
turns out to be cumbersome to deal with is the random permutation,
which got introduced by the first criterion, and to get a bound on the
actual state $\Eacc(\rho_{RS^N})$ instead; we show how to do this in
Section~\ref{sec:unpermute}.

We point out that the ``cheap'' way to deal with the random
permutation would be to simply modify the sampling protocol by
\emph{really} permuting the registers at the end of the protocol, so
that the permuted state {\em is} the final state after the sampling
protocol.  Besides being esthetically less appealing, because it would
mean a less natural and more complicated sampling protocol than really
necessary, this would also give more freedom to the party who chooses
the permutation in choosing it adversarially. For instance, in our
application in Section~\ref{sec:cointoss}, where the final state is
used to produce a high min-entropy source, we cannot allow that either
player can rearrange the registers and so, say, move the zero-outputs
into the positions he wants them to be.

\subsection{Permutation Invariance of our Sampling Protocols}

As a first step in analyzing the sampling protocol \purifbased of
Fig.~\ref{fig:samplingprot}, we show that it satisfies the above
definition of permutation invariance. Given that Sam's actions are
obviously symmetric with respect to permuting his registers, this is
probably not very surprising; spelling out the details though still
turns out to be somewhat cumbersome. We therefore move the proof to
Sect.~\ref{sec:purifscheme} and simply give a high-level proof
sketch below.

\begin{proposition}
  \label{prop:perminvofpurifsampling}
  The protocol \purifbased of Fig.~\ref{fig:samplingprot} satisfies
  Definition~\ref{def:perminv}.
\end{proposition}
\begin{proof}[sketch]
  For the first criterion, we need to argue that any adversary against
  the real sampling protocol can be adapted into an adversary against a
  symmetrized version of the protocol that will yield the same output
  state, up to a random permutation.  

  We first observe that when sampling from a permutation invariant
  operator, it doesn't matter which registers we sample from since the
  reduced density operator of any subset of $k$ registers is the same,
  i.e. $\rho_{S_t}= \rho_{S_{t'}}$ for any $t,t'\subseteq [N]$ of size
  $k$. Therefore we can make the simplifying assumption that we always
  sample from the first $k$ registers of $S^N$.

  We construct the symmetric adversary: from the symmetric state $\bar
  \rho_{P^NS^N}$ from the first criterion of
  Definition~\ref{def:perminv}, the adversary will compute the
  permutation $\pi\in \symgrp N$ applied on $S^N$. This permutation defines
  the set $t_\pi\subset[N]$ of positions to which $\pi$ sends
  positions $1,\dots, k$. The symmetric adversary will then simulate
  the real adversary on this sample $t_\pi$ and will permute the
  output according to $\pi$ before sending it to Sam (such that each
  register sent by the adversary aligns with the corresponding
  register on Sam's side).

  The second criterion follows from the observation that the maximal
  probability of measuring $\proj\varphi^{\otimes k}$ in the sampling
  protocol on input $\proj\theta^{\otimes N}$ is the fidelity between
  $\theta^{\otimes k}$ and $\varphi^{\otimes k}$ which is negligible
  in $k$ when $F(\theta_{ S}, \varphi_{ S})^2< 1-\epsilon$.

  The third criterion follows from the fact that the unsampled
  positions are untouched in both the real and the symmetrized
  protocols.\qed
\end{proof}

The following proposition allows us to apply the techniques of this
section to the LOCC sampling protocol presented in
Fig.~\ref{fig:EPRsampling}. Its proof can be found in
Sect.~\ref{sec:proofeprsampl}.

\begin{proposition}
  \label{prop:EPRsampling}
  The sampling protocol \eprlocc from Fig.~\ref{fig:EPRsampling}
  satisfies Definition~\ref{def:perminv}.
\end{proposition}
\begin{proof}[sketch]
  We need to argue that the protocol is permutation invariant in the
  sense of Definition~\ref{def:perminv}, and that it performs well on
  i.i.d. states. The first part follows from the permutation
  invariance of the choice of $t$ and $c$ and of the measurement on
  the sampler's qubits. Suppose Sam was to permute his register with
  $\pi\in \symgrp N$ before performing the sampling. Then we can
  modify the adversary such that it attacks the sampling protocol with
  this new ordering of Sam's register: if Sam chooses sample $t$,
  announce $\pi(t)$ to Paul instead, the same goes for $c$. Let $x$ be
  Paul's message to Sam, then permute $x$ such that it aligns
  correctly with the corresponding qubits on Sam's register.  The
  probability of accepting is exactly the same and the output of the
  protocol will be shuffled according to $\pi$'s action on the
  unsampled qubits.

  The second criterion follows from the fact that the only state that
  is perfectly correlated in both the computational and the diagonal
  bases is the EPR pair $\ket{\Phi^+}$. Therefore if all of Paul and
  Sam's measurement outcomes are perfectly correlated in the randomly
  chosen basis, it should hold that they shared states close to
  perfect EPR pairs. More precisely, if they share a state
  $\ket\theta^{\otimes N}$ where each $\theta$ has fidelity at most
  $1-\epsilon$ with $\ket{\Phi^+}$, then their outputs cannot be
  perfectly correlated in at least one of the bases, except with
  negligible probability.  The third criterion follows trivially from
  the fact that the unsampled qubits are not measured or acted upon.
  \qed
\end{proof}

\subsection{Proof of Sampling Against Symmetric Adversaries}
\label{sec:symmetricproof}

By considering sampling protocols that are permutation invariant in the
sense of Definition~\ref{def:perminv}, we can use the specific
properties of symmetric states to upper-bound the failure probability
of such protocols for symmetric adversaries (adversaries which prepare a
state $\ket{\bar \rho_{P^N S^N}}$ that lives in the symmetric subspace
$\sym^N(\hilbert_P\otimes \hilbert_S)$).

Lemma~\ref{lem:postsamplupperbound} below shows that since symmetric
states are approximated by a mixture of i.i.d. states, then the output
of the sampling executed on such a mixture is approximated by a
mixture of states i.i.d. in states that are close to the reference
state $\varphi$.

\begin{lemma}\label{lem:postsamplupperbound}
  Let $\Eacc$ be the output of a sampling protocol that satisfies
  Definition~\ref{def:perminv} and let $\rho_{RS^N}\in
  \density{\hilbert_R\otimes \hilbert_S^{\otimes N}}$. For any
  $\epsilon > 0$ there exists a subnormalized measure $d\theta_S$ on
  the set of mixed states $\theta_S \in \density{\hilbert_S}$ which
  satisfy $F(\theta_{ S}, \varphi_{ S})^2\geq 1-\epsilon$ and an
  operator $\tilde\sigma_{S^n}$ such that
  \begin{equation}\label{eq:postsamplupperbound}
    \frac 1{n!} \sum_{\pi\in \symgrp n} \pi_{S^n}\Eacc(\rho_{RS^N})
    \pi_{S^n}^\dagger \leq c_{N,d^2}\cdot \int \theta^{\otimes 
        n}_{S^n}d\theta_S + \tilde \sigma_{S^n}
  \end{equation}
  and $\|\tilde \sigma_{S^n}\|_1 \leq \exp(-\Omega(N))$, where
  $c_{N,d^2}$ is the dimension of $\sym^N(\hilbert_P\otimes \hilbert_S)$.
\end{lemma}

\begin{proof}
  By Definition~\ref{def:perminv}, there exists $\barEacc$ and
  $\rhobar\in \sym^N(\hilbert_P\otimes \hilbert_S)$ such
  that
  \begin{equation}
    \frac 1{n!}\sum_{\pi\in \symgrp n} \proj\pi_\Pi\otimes
      \pi_{S^n} \Eacc(\rho_{RS^N})
      \pi_{S^n}^\dagger=
      \barEacc(\rhobar)\enspace.
  \label{eq:f284hf}
  \end{equation}
  Therefore it suffices to prove the statement for $\bar{\cal
    E}^\acc_{P^N S^N\rightarrow S^n}$ obtained by tracing out the
  register $\Pi$ from the output of $\barEacc$.
  
  Since $\ket{\rhobar}\in\sym^N(\hilbert_P\otimes \hilbert_S)$, it
  holds by remark~\ref{rem:projsymsubspace} that $\rhobar\leq
  c_{N,d^2}\cdot \int \proj\theta^{\otimes N}_{P^N
    S^N}\;d\ket{\theta_{PS}}$ where $d\ket{\theta_{PS}}$ is the normalized
  Haar measure on the set of pure states on $\hilbert_P\otimes
  \hilbert_S$. It follows that
  \begin{align*}
    \bar{\cal
    E}^\acc_{P^N S^N\rightarrow S^n}(\rhobar)
    &\leq \bar{\cal
      E}^\acc_{P^N S^N\rightarrow S^n}
      \left(c_{N,d^2}\cdot \int \proj\theta^{\otimes N}_{P^N
      S^N}\;d\ket\theta\right)\\ 
    &=
      \begin{aligned}[t]
        c_{N,d^2}\cdot \bar{\cal E}^\acc_{P^N S^N\rightarrow S^n}
        \bigg(&\int_{\theta_{S}\approx^\epsilon \varphi_{S}}
        \proj\theta^{\otimes N}_{P^N S^N}\; d\ket\theta\\
        &+\int_{\theta_{S}\not\approx^\epsilon \varphi_{S}}
        \proj\theta^{\otimes N}_{P^N S^N}\; d\ket\theta\bigg)
      \end{aligned}
\\
    &\leq c_{N,d^2}\cdot \int_{\theta_{ S}\approx^\epsilon \varphi_{S}} 
      \theta^{\otimes n}_{S^n}\;d\theta_S
      + \tilde \sigma_{S^n} \\
  \end{align*}
  where $\theta_{ S}\approx^\epsilon \varphi_{ S}$ means that
  $F(\theta_{ S}, \varphi_{ S})^2\geq 1-\epsilon$ and where
  the operator $\tilde\sigma_{S^n}:= c_{N,d^2}\cdot\bar{\cal
    E}^\acc_{P^N S^N\rightarrow S^n} \left(\int_{\theta\not\approx^\epsilon \varphi}
    \proj\theta^{\otimes N} d\ket\theta\right)$ satisfies
  $\|\tilde\sigma_{S^n}\|_1 \leq \exp(-\Omega(N))$ by the second
  criterion of Definition~\ref{def:perminv}. The last inequality of
  the above follows from the third criterion of
  Definition~\ref{def:perminv} and from Remark~\ref{rem:fq09h4g}:
  since the trace non-increasing map $\bar{\cal
    E}^\acc_{P^N S^N\rightarrow S^n}$ does not act on the unsampled qubits, the
  state of $S^n$ \emph{after} the application of this map is
  upper-bounded by the state of the unsampled qubits \emph{before} its
  application. 

  Finally, the measure $d\theta_S$ is obtained by taking the partial
  trace over $P$ on the measure $d\ket{\theta_{PS}}$ on the restricted
  set of $\ket{\theta_{PS}}$ where $F(\theta_{ S}, \varphi_{
    S})^2\geq 1-\epsilon$. This corresponds to a measure proportional
  to the Hilbert-Schmidt measure~\cite{0305-4470-34-35-335,renner2010}
  over density operators on $\hilbert_S$ which have fidelity squared at
  least $1-\epsilon$ with $\varphi_S$.  \qed
\end{proof}

From the above Lemma, we can conclude that the \emph{permuted} output
of the sampling protocol is upper bounded by an ideal state in the
spirit of~\eqref{intuitmain}.

\begin{corollary}\label{cor:fja948h9hf}
  Let $\Eacc$ be the output of a sampling protocol that satisfies
  Definition~\ref{def:perminv} and let $\rho_{RS^N}\in
  \density{\hilbert_R\otimes \hilbert_S^{\otimes N}}$. For any
  $\epsilon>0$, there exist a subnormalized $\epsilon$-ideal operator
  $\psi_{S^n}\in {\cal D}_\leq(\hilbert_S^{\otimes n})$ and $\sigma_{S^n}$ such that
  \begin{equation}
    \label{eq:f0849b3493b9fb39bf}
    \frac 1{n!} \sum_{\pi\in \symgrp n} \pi_{S^n} \Eacc(\rho_{RS^N})
    \pi_{S^n}^\dagger \leq c_{N,d^2}\cdot\psi_{S^n}
    +\sigma_{S^n}
  \end{equation}
  where $\|\sigma_{S^n}\|_1 \leq \exp(-\Omega(N))$.
\end{corollary}

\begin{proof}
  Fix $\beta= \epsilon/2$ and let $d\theta_S$ and $\tilde\sigma_{S^n}$
  be as in Lemma~\ref{lem:postsamplupperbound} for parameter $\beta$,
  i.e. such that
    \begin{equation}\label{eq:f9849fwh3fh}
    \frac 1{n!} \sum_{\pi\in \symgrp n} \pi_{S^n}\Eacc(\rho_{RS^N})
    \pi_{S^n}^\dagger \leq c_{N,d^2}\cdot \int \theta^{\otimes 
        n}_{S^n}d\theta_S + \tilde \sigma_{S^n}
    \end{equation}
    where $d\theta_S$ is a subnormalized measure on the set of mixed
    states which satisfy $F(\theta_S,\varphi_S)^2\geq 1-\beta$ and
    where $\tilde \sigma_{S^n}$ has negligible norm.

    Let $\tau_{P^nS^n}:= \int \proj \theta_{P^nS^n}^{\otimes
      n}d\theta_S$ be an extension of $\int \theta_{ S^n}^{\otimes
      n}d\theta_S$ where each $\ket{\theta_{ P S}}$ is such that
    $|\braket{\theta_{ P S}}{\varphi_{ P S}}|^2= F(\theta_{ S},
    \varphi_{ S})^2 \geq 1-\beta$ and let $ \tilde \sigma_{P^nS^n}$ be
    an extension of $\tilde \sigma_{S^n}$. Then from
    Lemma~\ref{lem:puriftauisideal}, we have
  \begin{equation}
    \label{eq:f48g93b9gb3g}
    \trace {(\id-\mathbb P_{P^n S^n}^{2\beta n,
        \ket\varphi})\left(\tau_{P^nS^n}\right)}  
    \leq \exp(-2\beta^2 n)\enspace.
  \end{equation}

  Choose $\psi_{S^n}= \trace[P^n]{\mathbb P_{P^n S^n}^{2\beta n,
      \ket\varphi}\tau_{P^nS^n}\mathbb P_{P^n S^n}^{2\beta n,
      \ket\varphi}}$. Then, using (\ref{eq:f9849fwh3fh}), we have
  \begin{align*}
    \frac 1{n!} \sum_{\pi\in \symgrp n}& \pi_{S^n}\Eacc(\rho_{RS^N})
    \pi_{S^n}^\dagger
    \leq c_{N,d^2}\cdot \int \theta^{\otimes 
      n}_{S^n}d\theta_S + \tilde \sigma_{S^n}\\
    &=  \trace[P^n]{c_{N,d^2}\cdot\tau_{P^nS^n} + \tilde \sigma_{P^nS^n}}= c_{N,d^2}\cdot\psi_{S^n} + \sigma_{S^n}
  \end{align*}
  where $\sigma_{S^n}:= \trace[P^n]{c_{N,d^2}(\tau_{P^nS^n} - \mathbb
    P_{P^n S^n}^{2\beta n, \ket\varphi}\tau_{P^nS^n}\mathbb P^{2\beta
      n, \ket\varphi})+ \tilde \sigma_{P^nS^n}}$ has norm
  upper bounded by
  \begin{equation*}
\|\sigma_{P^nS^n}\|_1 \leq
  c_{N,d^2}\|\tau_{P^nS^n} - \mathbb P_{P^n S^n}^{2\beta n,
    \ket\varphi}\;\tau_{P^nS^n}\;\mathbb P_{P^n S^n}^{2\beta n,
    \ket\varphi}\|_1 + \|\tilde\sigma_{P^nS^n}\|_1\leq
  \exp(-\Omega(N))
\end{equation*}
by first applying the triangle inequality and then
  the Gentle Measurement's Lemma~\cite{winter99,on02} with the bound
  of (\ref{eq:f48g93b9gb3g}).\qed
\end{proof}
It should be noted that the operator $ \sigma_{S^n}$ from the above
Corollary is not positive semidefinite in general, but since its norm
is negligible, this shouldn't matter because it can simply be ignored
for most applications.

\subsection{Proof Against Arbitrary Adversaries: Unpermuting the
  Output}
\label{sec:unpermute}

In order to conclude that the sampling protocol works as intended on
an arbitrary input state and adversarial strategy, we need to argue
that if we remove the permutation from the contents
of~(\ref{eq:f0849b3493b9fb39bf}), then the left-hand side, which
becomes the post-sampling state, is still approximated by a state
having a purification in a low-error subspace. It turns out that the
intuitive statement ``if the permuted output is ideal then the
non-permuted output is also ideal'' that we want to show is quite
tricky to prove. We stress that this step is necessary if we want to
keep the permutation ``under the hood'' and have a statement that
doesn't require to physically shuffle the systems, which would lead to
unnatural sampling protocols.

Lemma~\ref{lem:unpermuteideal} below is the first step in this proof,
it shows that the property of having a purification in a low-error
subspace, i.e. of being \emph{ideal}, does indeed persist after
``unpermutation'' of the registers. Its proof is straightforward and
can be found in Appendix~\ref{app:addproofs2}.
\begin{lemma}\label{lem:unpermuteideal}
  Let $\epsilon>0$ and let $\sigma_{S^n}\in
  \density{\hilbert_S^{\otimes n}}$ be such that $\frac
  1{n!}\sum_{\pi\in \symgrp n} \pi_{S^n} \sigma_{S^n}\pi^\dagger_{S^n}$ is
  $\epsilon$-ideal
  , then $\sigma_{S^n}$ 
  is also $\epsilon$-ideal.
\end{lemma}

We now have all the tools we need to prove our main result,
Theorem~\ref{thm:mainresultunpermuted} below. Its proof combines the
above lemma with Lemmas~\ref{lem:puriftauisideal}
and~\ref{lem:postsamplupperbound} to show that the output of the
sampling is negligibly close to a state that is post-selected from a purification of
an ideal state.

\begin{theorem}[Main Result]
  \label{thm:mainresultunpermuted}
  Let $\Eacc$ be the output of a sampling protocol that satisfies
  Definition~\ref{def:perminv} and let $\rho_{RS^N}\in
  \density{\hilbert_R\otimes \hilbert_S^{\otimes N}}$. For any
  $\epsilon>0$, there exists a non-normalized vector
  \begin{equation*}
    \ket{\tilde \psi_{{R'}P^nS^n}}\in \hilbert_{R'}\otimes \Delta_{\epsilon n}(\ket
    \varphi^{\otimes n}_{P^nS^n})
  \end{equation*}
  and a completely positive trace non-increasing superoperator
  $\tilde {\cal K}_{{R'}P^n\rightarrow \complex}$ such that
\begin{equation*}
  \left\| 
     \Eacc(\rho_{RS^N}) - c_{N,d^2}
    (\tilde {\cal K}_{ {R'}P^n}\otimes \idcptp_{S^n})( \tilde \psi_{ {R'}P^nS^n})\right\|_1\leq
  \exp(-\Omega(N))
  \end{equation*}
\end{theorem}

By means of Proposition~\ref{prop:leqequivpostsel} and Remark~\ref{rem:fq09h4g}, we can express the statement of Theorem~\ref{thm:mainresultunpermuted} in terms of an operator inequality as suggested in (\ref{intuitmain}), rather than by means of post-selection. 

\begin{corollary}
  \label{thm:mrupperbound}
  Let $\Eacc$ be the output of a sampling protocol that satisfies
  Definition~\ref{def:perminv} and let $\rho_{RS^N}\in
  \density{\hilbert_R\otimes \hilbert_S^{\otimes N}}$. For any
  $\epsilon>0$, there exist a subnormalized $\epsilon$-ideal operator
  $\psi_{S^n}\in {\cal D}_\leq(\hilbert_S^{\otimes n})$ and
  $\sigma_{S^n}$ such that
  \begin{equation*}
    \Eacc(\rho_{RS^N}) \leq c_{N,d^2}\cdot\psi_{S^n}
    +\sigma_{S^n}
  \end{equation*}
  where $\|\sigma_{S^n}\|_1 \leq \exp(-\Omega(N))$.
\end{corollary}

\begin{proof}[of Theorem~\ref{thm:mainresultunpermuted}]
  Let $\psi_{S^n}$ and $ \sigma_{S^n}$ be as in the statement of
  Corollary~\ref{cor:fja948h9hf}, i.e. such that
  \begin{equation}
    \label{eq:fjowiejfq4g}
    \frac 1{n!} \sum_{\pi\in \symgrp n} \pi_{S^n}\Eacc(\rho_{RS^N})
    \pi_{S^n}^\dagger \leq c_{N,d^2}\cdot\psi_{S^n}
    +\sigma_{S^n}
  \end{equation}
  and define $ \tau_{S^n}:= \psi_{S^n} + c_{N,d^2}^{-1}\cdot
  \sigma_{S^n}$. Since $\psi_{S^n}$ is $\epsilon$-ideal, let
  $\ket{\psi_{{R'}P^nS^n}}$ be the purification of $\psi_{S^n}$ that
  lives in the low error subset $\hilbert_{R'}\otimes \Delta_{\epsilon
    n}(\ket\varphi^{\otimes n}_{P^nS^n})$. Let
  $\ket{\tau_{{R'}P^nS^n}}$ be a purification\footnote{The existence
    of a purification of $\tau_{S^n}$ with this property can be argued
    by using Uhlmann's Theorem: since $\tau_{S^n}$ is close in
    fidelity to $\psi_{S^n}$, for any purification
    $\ket{\psi_{{R'}P^nS^n}}$ of $\psi_{S^n}$, there exists a
    purification $\ket{\tau_{{R'}P^nS^n}}$ that is also close to
    $\ket{\psi_{{R'}P^nS^n}}$. } of $\tau_{S^n}$ such that $\|
  \psi_{{R'}P^nS^n}- \tau_{{R'}P^nS^n}\|_1\leq \exp(-\Omega(N))$.
  From (\ref{eq:fjowiejfq4g}) and
  Proposition~\ref{prop:leqequivpostsel} we can show that there exists
  a trace non-increasing completely positive map ${\cal
    K}_{{R'}P^n\rightarrow \Pi}$ that produces a classical register
  $\Pi$ from purification registers ${R'}P^n$ with the property that
  \begin{align*}
    \frac 1{n!} \sum_{\pi\in \symgrp n} \proj \pi_\Pi \otimes \pi_{S^n}
    \Eacc(\rho_{RS^N})\pi^\dagger_{S^n}
    &= c_{N,d^2}({\cal K}_{{R'}P^n\rightarrow \Pi}\otimes \idcptp_{S^n}) (
      \tau_{{R'}P^nS^n})\enspace. 
  \end{align*}

  Suppose now we were to submit both sides of the above equality to
  the following quantum operation: measure register $\Pi$ and undo the
  observed permutation on register $S^n$. The left-hand side of the
  above would become $\Eacc(\rho_{RS^N})$
  whereas the right-hand side becomes
  \begin{equation*}\label{eq:f184h19f}
    c_{N,d^2}\cdot\sum_{\pi\in \symgrp n} (\bra \pi_\Pi\otimes \pi^{-1}_{S^n})({\cal
        K}_{{R'}P^n\rightarrow \Pi}\otimes \idcptp_{S^n}) ( 
      \tau_{{R'}P^nS^n})(\ket \pi_\Pi\otimes
      (\pi^{-1}_{S^n})^\dagger)\enspace .
  \end{equation*}

  We now show how to represent this operator in a way that corresponds
  to the statement we need to prove, i.e. as post-selected from a
  rank-one operator living almost entirely in the low-error
  subspace. To this end, define\footnote{It is always possible to
    define such an isometry and projector for any trace non-increasing
    completely positive superoperator ${\cal E}_{A\rightarrow B}$. To
    see this, let ${\cal E}(\sigma_A)=\sum_k E_k \sigma_A E_k^\dagger$
    where $E_k\in L(\hilbert_A,\hilbert_B)$ are the Kraus operators of
    $\cal E$ and define the isometry $U_{A\rightarrow BZ}$ as mapping
    an arbitrary state $\ket\psi_A$ to $\sum_k E_k \ket\psi_A \ket k_Z
    + \sqrt{ \id-\sum_k E_k^\dagger E_k} \ket \psi_A \ket \bot_Z$
    where $\ket\bot_Z$ is orthogonal to $\ket k_Z$ for every $k$. Then
    $\mathbb P_Z=\sum_k \proj k_Z$ suffices as the required projector
    since $\trace[Z]{(\id_B\otimes \mathbb P_Z) U_{A\rightarrow BZ}
      \sigma_A U_{A\rightarrow BZ}^\dagger}= \sum_k E_k \sigma_A
    E_k^\dagger={\cal E}_{A\rightarrow B}(\sigma_A)$. } an isometry
  $U_{{R'}P^n\rightarrow Z\Pi}$ that purifies the action of ${\cal
    K}_{{R'}P^n\rightarrow \Pi}$, i.e. such that for any
  $\nu_{{R'}P^n}$,
  \begin{equation*}
    \mathcal{K}_{{R'}P^n\rightarrow \Pi} (\nu_{{R'}P^n}):=
    \trace[Z]{(\mathbb P_{Z}\otimes \id_{\Pi}) \cdot U_{{R'}P^n\rightarrow
        Z\Pi} \cdot \nu_{{R'}P^n}\cdot
      (U_{{R'}P^n\rightarrow Z\Pi})^\dagger}
  \end{equation*}
  for some projector $\mathbb P_{Z}$.  Using this representation, the
  post-sampling operator can be expressed as
  \begin{equation}\label{eq:fqowiefq04g}
   \Eacc(\rho_{RS^N})=
    c_{N,d^2}\cdot\trace[Z]{(\mathbb P_{Z}\otimes \id_{ S^n}) \cdot
      \sum_{\pi\in \symgrp n}[U^\pi_{{R'}P^n\rightarrow 
        Z}\otimes \pi^{-1}_{S^n}]   (\tau_{{R'}P^nS^n})
    }
  \end{equation}
  where $U^\pi_{{R'}P^n\rightarrow Z}:= (\id_Z\otimes \bra
  \pi_\Pi)\cdot U_{{R'}P^n\rightarrow Z\Pi}$ and where $[U](\rho)$ is
  short for $U\rho U^\dagger$.

  Define the operator
  \begin{equation*}
    \tilde\psi_{ZS^n}:= \sum_{\pi\in \symgrp n}(U^\pi_{{R'}P^n\rightarrow
      Z}\otimes \pi^{-1}_{S^n}) \psi_{{R'}P^nS^n}
    (U^\pi_{{R'}P^n\rightarrow Z}\otimes \pi^{-1}_{S^n})^\dagger\enspace. 
  \end{equation*}
  where $ \psi_{{R'}P^nS^n}$ is the purification of $ \psi_{S^n}$
  defined earlier. It isn't too hard to show that $\tilde\psi_{S^n}$
  is such that $\psi_{S^n}= \frac 1{n!}\sum_{\pi\in \symgrp
    n}\pi_{S^n} \tilde\psi_{S^n}\pi_{S^n}^\dagger$.  Since
  $\psi_{S^n}$ has a purification in the low-error subspace,
  Lemma~\ref{lem:unpermuteideal} implies that $\tilde\psi_{S^n}$
  itself admits a purification in this subspace. Let
  $\ket{\tilde\psi_{{R'}P^nS^n}}$ be this purification and let $\tilde
  {\cal K}_{{R'}P^n\rightarrow \complex}$ be the superoperator that
  first maps $\ket{\tilde \psi_{{R'}P^nS^n}}$ to $\tilde \psi_{ZS^n}$
  and then applies $\sigma_Z\mapsto\trace[Z]{\mathbb P_Z\sigma_Z}$ to
  register $Z$. Then, using the definition of $\tilde \psi_{
    {R'}P^nS^n}$ and $\tilde {\cal K}_{{R'}P^n}$, and since completely
  positive trace non-increasing maps cannot increase the trace
  distance,
  \begin{align*}
    \|\Eacc(\rho_{RS^N}) &- c_{N,d^2}
      (\tilde {\cal K}_{ {R'}P^n}\otimes \idcptp_{S^n})(\tilde \psi_{
      {R'}P^nS^n})\|_1\\
    &=
      \begin{aligned}[t]
        \bigg\|c_{N,d^2}\cdot&\tr_{Z}\bigg(\mathbb P_{Z}\otimes \id_{
          S^n}) \cdot \\&\sum_{\pi\in \symgrp
          n}[U^\pi_{{R'}P^n\rightarrow Z}\otimes \pi^{-1}_{S^n}]
        \Big(\tau_{{R'}P^nS^n}-\psi_{{R'}P^nS^n}\Big)
        \bigg)\bigg\|_1
      \end{aligned}
\\
    &\leq c_{N,d^2}\cdot\| \tau_{{R'}P^nS^n} - \psi_{{R'}P^nS^n}\|_1\\
    &\leq \exp(-\Omega(N))
  \end{align*}
  where in the first inequality $\Eacc(\rho_{RS^N})$ is replaced with (\ref{eq:fqowiefq04g}) and the
  last inequality follows from our choice of $\ket
  {\tau_{{R'}P^nS^n}}$. \qed

\end{proof}


%% file: sec-cointoss.tex
\subsection{The Protocol}
\label{sec:cointossprotocol}

The protocol for 
randomness generation is depicted in Fig.~\ref{fig:cointossprotocol}. 
The protocol works as follows: Alice first has to generate $N$ EPR
pairs and send half of each to Bob. Bob then uses our sampling
protocol of Fig.~\ref{fig:samplingprot} to certify that the state
Alice sent him is (close to) the prescribed state. If Bob's check
succeeds, then our quantum sampling result says that Alice basically
prepared the right state, up to a few errors. Bob's measurement
outcome will then have very high min-entropy (arbitrarily close to the
maximum $n$).

\begin{figure}[h]
  \begin{framed}
    \begin{enumerate}
    \item Alice prepares the state $\ket{\Phi^+}^{\otimes N}_{A^NB^N}$
      for $\ket{\Phi^+}:= \frac 1{\sqrt 2}(\ket {00}+ \ket{11})$ and
      sends the system $B^N$ to Bob.
    \item \label{step:samplingforct}Alice and Bob perform protocol
      \purifbased from Fig.~\ref{fig:samplingprot} with Alice as the
      prover and Bob as the sampler and with $k=\beta N$ for $\beta>0$
      such that $\beta N$ is an integer.  Let $\rho_{A^nB^n}\in
      \density{(\hilbert_2\otimes \hilbert_2)^{\otimes n}}$ be the
      resulting normalized joint state of $n= N-k$ pairs of
      qubits.
      \item Alice and Bob respectively measure their $n$ qubits in the
        computational basis and output their
        respective measurement outcomes $X_A$ and $X_B$.
    \end{enumerate}
  \end{framed}
  
  \caption{The randomness generation protocol. $N$ is the security
    parameter, $\beta$ determines the size of the sample.}
  \label{fig:cointossprotocol}
\end{figure}

\subsection{Entropy of Alice and Bob's Outputs}
\label{sec:cointossproof}

Since Alice is the preparer of the $N$ EPR pairs, her output will have
high min-entropy. The tricky part of the following proof is showing
that Bob's freedom in choosing $t$ and accepting or refusing the
sampling outcome cannot influence too much the distribution of Alice's
measurement outcome.

\begin{lemma}[Entropy of Alice's output]
\label{lem:aliceentropy}
  If Alice follows the protocol, then for any $\gamma>0$, her output
  $X_A\in \bool^n$ satisfies
  \begin{equation*}
    \hmin{X_A}\geq (1-\gamma)n\enspace ,
  \end{equation*}
  except with probability negligible in $n$.
\end{lemma}
\begin{proof}
  Let $\rho_{A^NB^N}$ be the joint state of Alice and Bob before the
  sampling phase. As the preparer of the quantum state, Alice prepares
  $N$ perfect EPR pairs (i.e. $\rho_{A^NB^N}= \proj{\Phi^+}^{\otimes
    N}$), so her measurement outcome would have maximal min-entropy
  for the $n$ remaining qubits were it not for Bob's actions. Bob can
  bias the outcome of Alice's measurement in two possible ways: (1) he
  can measure his register $B^N$ \emph{before} choosing $t$ and make
  $t$ depend on this measurement outcome and (2) he can make the
  sampling abort even though Alice was honest. We analyze both
  possibilities separately, showing that each cannot reduce the
  min-entropy by more than a small linear amount, except with
  negligible probability.

  For (1), suppose Bob performs some measurement on his register $B^N$
  that yields sample choice $t\subset[N]$ with probability $p_t$ and
  results in the reduced density operator $\rho^t_{A^N}$ on Alice's
  side. Suppose also that Alice was to measure her whole state at this
  point, resulting in a measurement outcome $X_A\in \bool^N$. Observe
  that by the law of total probability, 
  \begin{equation*}
    2^{-N}=  2^{-\hmin{X_A}_{\rho}}= \sum_t p_t\cdot 2^{-\hmin{X_A\mid
        T=t}_{\rho^t}} 
    \enspace ,  
  \end{equation*}
  where $2^{-\hmin{X_A\mid T=t}_{\rho^t}}$ gives the maximal
  probability of guessing $X_A$ given $T=t$ when $X_A$ was obtained by
  measuring $\rho^t_{A^N}$.  It holds by Markov's inequality that
  \begin{equation*}
    \sum_t p_t\cdot [\hmin{X_A\mid T=t}_{\rho^t} \leq N-(\alpha
    N)]\leq 2^{-\alpha N}
  \end{equation*}
  where $[\cdot ]$ is the Iverson bracket which evaluates to 1 if the
  contents is true and to 0 otherwise.  In other words, the values of
  $t$ for which $\hmin{X_A\mid T=t}_{\rho^t}$ is less than $(1-\alpha)
  N$ have combined probability less than $ 2^{-\alpha N}$.  Now, Alice
  does not measure her whole state, but instead only those positions
  that do not belong to $t$, so let $X_A^{\bar t}$ be the outcome of
  measuring the qubits outside of $t$ and let $X_A^t$ be the outcome
  for the positions in $t$.  The following holds except with
  negligible probability over the choice of $t$:
  \begin{equation}\label{eq:fjq0293hf0q3hg}
    \hmin{X_A^{\bar t}\mid T=t}\geq \hmin{X_A\mid T=t, X^t_A}\geq
    (1-\alpha -\beta)N 
  \end{equation}
  where the last inequality follows from the chain rule for the
  min-entropy with $\hmax{X_A^t}= \beta N$.
  
  To deal with (2), observe that
  \begin{equation}\label{eq:FHAEfh0wh0fah}
    2^{-\hmin{X_A^{\bar t}\mid T=t, \acc}} \leq 2^{-\hmin{X_A^{\bar
      t}\mid T=t}}/\Pr[\acc] \leq 2^{- \hmin{X_A^{\bar t}\mid T=t}+\alpha N}
  \end{equation}
  whenever $\Pr[\acc]\geq 2^{-\alpha N}$. 

  We can conclude that, except with negligible probability
  upper bounded by $2\cdot 2^{-\alpha N}$, the min-entropy of Alice's
  output is
  \begin{equation*}
    \hmin{X_A^{\bar t}\mid T=t, \acc}\geq  (1-2\alpha -\beta)N 
  \end{equation*}
  by combining the bounds (\ref{eq:fjq0293hf0q3hg}) and
  (\ref{eq:FHAEfh0wh0fah}) and the respective probabilities that these
  bounds hold.  The statement is satisfied by choosing $\alpha$ and
  $\beta$ such that $\gamma= 2\alpha + \beta$ and noting that
  $N>n$.\qed
\end{proof}

We rely on the next Lemma to lower-bound the amount of min-entropy in
the measurement outcome of Bob. It says that if the joint state of
Alice and Bob lives in a quantum Hamming ball of small radius around
$n$ copies of an EPR pair, then Bob's reduced density operator has
high min-entropy.
\begin{lemma}\label{lem:minentropyidealstate}
  Let $\epsilon>0$ and $\ket{\sigma_{RP^nS^n}}\in \hilbert_R\otimes
  \Delta_{\epsilon n}(\ket{\Phi^+}^{\otimes n}_{P^nS^n})$. It holds that
  \begin{equation*}
\hmin{S^n}_\sigma
  \geq (1-\epsilon-h(\epsilon))n \enspace .
\end{equation*}
\end{lemma}
\begin{proof}
  Let $\Pi_\epsilon= \{E\subseteq[n]: |E|\leq \epsilon n\}$ and let
  $\mathbb P_{P^nS^n}^{\epsilon n, \ket{\Phi^+}} = \sum_{E\in
    \Pi_\epsilon} \mathbb P^E_{P^nS^n} $ be the projector onto
  $\Delta_{\epsilon n}(\ket{\Phi^+}^{\otimes n}_{P^nS^n})$ where
  \begin{equation*}
    \mathbb P^E_{P^nS^n}= \bigotimes_{i\in E} (\id-\proj{\Phi^+})_{P_iS_i}
      \bigotimes_{i\notin E} 
      \proj {\Phi^+}_{P_iS_i}\enspace .
  \end{equation*}
  Define $\ket{\tilde \sigma^E_{RP^nS^n}}= (\id_R\otimes \mathbb P^E_{P^nS^n})
  \ket{\sigma_{RP^nS^n}}$ for each $E\in \Pi_\epsilon$. It holds by
  Proposition~\ref{prop:smallnumbterms} that
  \begin{equation*}
    {\sigma_{RP^nS^n}}= \sum_{E,E'\in\Pi_\epsilon}\sketbra
    {\tilde \sigma^E_{RP^nS^n}}{\tilde \sigma^{E'}_{RP^nS^n}}  
    \leq 2^{h(\epsilon)n} \sum_{E\in\Pi_\epsilon}
    \proj {\tilde \sigma^E_{RP^nS^n}}
  \end{equation*}
  because the set $\Pi_\epsilon$ contains at most
  $ 2^{h(\epsilon)n}$ elements.  Furthermore, we know by the
  definition of $\ket{\tilde \sigma_{RP^nS^n}^E}$ that
  \begin{equation*}
    \frac{\tilde\sigma^E_{S^n}}{\|\tilde\sigma^E_{S^n}\|_1}= 
    \left(\bigotimes_{i\notin E} 
      \frac {\id_{S_i}}2\right)\otimes \psi_{S_E }\leq 2^{-n+|E|}\id_{S^n}
  \end{equation*}
  for some normalized state $ \psi_{S_E}$ living on register $S_E=
  \bigotimes_{i\in E} S_i$. Since $|E|\leq \epsilon n$, it directly
  follows that
  \begin{equation*}
    \sigma_{S^n}\leq 2^{h(\epsilon)n} \sum_{E\in \Pi_\epsilon} 
    \tilde \sigma^E_{S^n}\leq  2^{-(1-\epsilon-h(\epsilon))n}\id_{S^n}
  \end{equation*}
  and we can thus conclude that $\hmin{S^n}_\sigma \geq
  (1-\epsilon-h(\epsilon))n $.  \qed
\end{proof}

Lower-bounding Bob's output min-entropy is essentially applying
Lemma~\ref{lem:minentropyidealstate} to Bob's state after the sampling
step of protocol of Fig.~\ref{fig:cointossprotocol} which can be
approximated by an ideal state by means of our main result
(Corollary~\ref{thm:mrupperbound}).
\begin{lemma}[Entropy of Bob's output] 
  If Bob follows the protocol, for any $\gamma>0$, his output $X_B\in
  \bool^n$ satisfies
  \begin{equation*}
    \hmin{X_B}\geq (1-\gamma)n\enspace ,
  \end{equation*}
  except with probability negligible in $n$.
\end{lemma}
\begin{proof}
  The security of the protocol against dishonest Alice is almost a
  direct consequence of our quantum sampling result
  (Theorem~\ref{thm:mainresultunpermuted}). Let $\rho_{B^n}\in\mathcal
  D(\hilbert_2^{\otimes n})$ be the normalized state of Bob after
  step~\ref{step:samplingforct} of the protocol of
  Fig.~\ref{fig:cointossprotocol} given that Bob did not reject and
  let $P_\acc$ be the probability that he did not reject the sampling. By
  Corollary~\ref{thm:mrupperbound}, it holds that for any $\epsilon>0$
  there exists an ideal $\psi_{B^n}$ and an operator $\sigma_{B^n}$
  with negligible norm such that
  \begin{equation}\label{eq:08h2f4fh20hf}
    \rho_{B^n}\leq P_\acc^{-1}( c_{N,d^2} \psi_{B^n} +
    \sigma_{B^n})\enspace .
  \end{equation}
  Let $\tilde\psi_{B^n}=\frac{c_{N,d^2}}{P_\acc} \cdot\psi_{B^n}$. Then
  \begin{equation*}
  \left\| \frac{c_{N,d^2}}{P_\acc} (\psi_{B^n} +
  \sigma_{B^n}) - \tilde \psi_{B^n}\right\|_1=
  \frac{1}{P_\acc}\|\sigma_{B^n}\|_1\enspace ,
\end{equation*}
 which is negligible in $N$
  whenever $P_\acc$ is non-negligible. It follows that except with
  negligible probability, the right-hand side of
  (\ref{eq:08h2f4fh20hf}) will behave exactly like $\tilde
  \psi_{B^n}$, in which case their min-entropy will be
  equal. 
  This min-entropy is bounded below by
  \begin{equation}\label{eq:fh034h0h0h0h0}
    \hmin{\tilde \psi_{B^n}} = \hmin{\psi_{B^n}}  -\log \frac{c_{N,d^2}}{P_\acc} \geq (1-\epsilon -
    h(\epsilon))n -\log \frac{c_{N,d^2}}{P_\acc}
  \end{equation}
  by Lemma~\ref{lem:minentropyidealstate}. 

  Using the bound of (\ref{eq:fh034h0h0h0h0}), we can claim that the
  min-entropy of $\rho_{B^n}$ is lower-bounded by
  \begin{equation*}
    (1-\epsilon - h(\epsilon)-\alpha)n
  \end{equation*}
  unless one of two negligible probability events occurred. The first
  event is that $\rho_{B^n}$ behaves like $\sigma_{B^n}$ instead of
  $\tilde \psi_{B^n}$ and the second event is that Bob accepted the
  outcome of a sampling that had probability $P_\acc\leq
  c_{N,d^2}\cdot 2^{-\alpha n}$ of being accepted. We can conclude
  that the result $X_B$ of measuring $\rho_{B^n}$ in the computational
  basis will have min-entropy at least $(1-\epsilon -
  h(\epsilon)-\alpha)n$, except with negligible probability.  The
  statement follows by choosing $\epsilon$ and $\alpha$ in the above
  such that $\gamma= \epsilon + h(\epsilon)+\alpha$.\qed

\end{proof}


%% file: sec-conclusion.tex
\section{Conclusion and Open Questions}

Statistical sampling is a natural task that is well understood from a
classical perspective. Classical tools such as Hoeffding's inequality,
Azuma's inequality and other results on concentration of measure that
are used to analyze classical sampling (and quantum sampling to a
certain degree~\cite{bouman-fehr}) are of no use when trying to sample
from quantum data with a \emph{mixed} reference state.  The tools of
symmetric invariance can substitute the classical tools up to a
certain degree when analyzing fully quantum sampling protocols. We
have introduced a framework for sampling mixed states by presenting a
general sampling protocol and we have shown that if an instantiation
of that general protocol respects simple criteria, then it can be used
to certify that a quantum population is close to an $n$-fold tensor
product of a reference state $\varphi$ in an adversarial setting.

Sampling of a quantum population is a new concept and many questions
are left unanswered, especially when sampling with a mixed reference
state where the usual (classical) tools do not apply. Precisely,
future directions for this work include:
\begin{enumerate}
\item A formulation of our results where a conclusion can be made when
  an error rate significantly larger than 0 has been observed. From an
  observed error rate of $\delta>0$ within the sample, we would want
  to conclude that the state of the remaining positions can be
  controlled by means of an $(\varepsilon+\delta)$-ideal state for small
  $\varepsilon>0$.
\item An extension of our results to multiple reference states for the
  same population instead of a fixed reference state $\varphi$, e.g.
  with reference states $\varphi_0, \varphi_1$ where register $i$ of
  the population is tested against $\varphi_{x_i}$ for $x\in
  \bool^n$. While sampling according to an arbitrary (pure) reference
  state is given ``for free'' for pure state sampling (since all pure
  states are related by a unitary transformation on the sampler's
  register), it requires more work in the case of mixed state
  sampling.
\item On top of the previous point, it is often useful for quantum sampling applications
   to have a statement in terms of an
  \emph{adaptive} sampling protocol where the reference states
  (i.e. the bits of $x$) are chosen adaptively by the adversary based
  on what positions were sampled. Such an extension would have
  applications in two-party cryptography where sampling is done in a
  sequential  manner using a 1- or 2-bit cryptographic primitive, such
  as cut-and-choose.  
   In fact, if our results were extended in such a
  way, it would allow to certify states with a 2-bit description (such
  as the BB84 encoding) using a 1-bit cut-and-choose, a task that is
  not known to be possible
  relying on existing sampling tools.
  The pure-state sampling framework of~\cite{bouman-fehr}
  was shown to apply in the adaptive setting in~\cite{fkszz13}. 
\end{enumerate}


%% file: app-proofs.tex
\subsection{Proof of Proposition~\ref{prop:perminvofpurifsampling}
}
\label{sec:purifscheme}

\newcommand{\Ut}{U_{R\rightarrow R'P^k}^t}
\newcommand{\Upi}{U_{R\rightarrow R'P^k}^{t_\pi}}
\newcommand{\Vt}{V^t_{S^N\rightarrow S^n S^k}}
\newcommand{\Vk}{V^{[k]}_{S^N\rightarrow S^n S^k}}
\newcommand{\Vpi}{V^{t_\pi}_{S^N\rightarrow S^n S^k}} We can assume
w.l.o.g. that the state $\rho_{RS^N}\in \density{\hilbert_R\otimes
  \hilbert_S^{\otimes N}}$ is pure and that adversarial strategies
against the protocol depicted in Fig.~\ref{fig:samplingprot} is
described by a family of isometries of the form $\Ut$ for $t\subseteq
[N]$ of size $k$, where $P^k$ represents the register sent to Sam and
supposed to contain the purifications of $\varphi_S$, and $R'$ is a
register kept by Paul.

For convenience, define the isometry $\Vt$
that, for any $t\subseteq [N]$, maps subsystems $S_i$ for $i\in t$
into the last $k$ subsystems (denoted $S^k$) and subsystems $S_i$ for
$i\notin t$ into the first $n=N-k$ subsystems (denoted $S^n$).  In
other words, isometry $V^t_{S}$ simply groups together the registers
to be sampled.

For an adversarial strategy as described above, the completely
positive trace non-increasing map $\Eacc$ that maps the input state
$\rho_{RS^N}$ to the sampler's conditional output is defined by
\begin{align*}
  \Eacc(\rho_{RS^N}):= \frac 1{\binom Nk}\sum_{t\subseteq [N]}
  \trace[R']{\bra{\varphi}^{\otimes k}_{P^kS^k}\cdot
    [U^t_R\otimes V^t_{S^N}](\rho_{PS})
    \cdot\ket{\varphi}^{\otimes 
      k}_{ P^k S^k}} \;.
\end{align*}
where we left the identity operator acting on $R'S^n$ implicit and
where $[U](\rho)$ is short for $U\rho U^\dagger$ for any isometry $U$.

The following property of $\Vt$ will be
useful for proving Lemma~\ref{lem:eqsymadv} below.
\begin{remark}\label{rem:dividepi}
  Let $\pi\in \symgrp N$, and let $t_\pi=  \{\pi^{-1}(i)\mid
  i\in [k]\}$. There exist $\tau^\pi\in \symgrp k$ and $\bar \tau^\pi\in
  \symgrp n$ such that $\Vk\cdot\pi_S = (
  \bar \tau^\pi_{S^n}\otimes \tau^\pi_{S^k} )\cdot
  \Vpi$. Furthermore, there is a
  one-to-one correspondence between permutations $\pi\in \symgrp N$ and
  triplets $(t_\pi, \tau^\pi, \bar\tau^\pi)$.
\end{remark}

\begin{lemma}\label{lem:eqsymadv}
  Protocol \purifbased from Fig.~\ref{fig:samplingprot} satisfies
  the first criterion of Definition~\ref{def:perminv}.

\end{lemma}
\begin{proof}
  We need to show the existence of a completely positive trace
  non-increasing map $\barEacc$ such that for any $\rho_{RS^N}$,
  \begin{equation}
    \label{eq:4rccr248cr9284cr2938}
    \frac 1{n!}\sum_{\pi\in \symgrp n} \proj\pi_\Pi\otimes
    \pi_{S'} \Eacc(\rho_{RS^N})
    \pi_{S'}^\dagger=
    \barEacc(\rhobar)
  \end{equation}
  for some symmetric purification $\ket{\rhobar}$ of $ \frac
  1{N!}\sum_{\pi\in \symgrp N} \pi_{S^N} \rho_{S^N} \pi_{S^N}^\dagger$ where
  $\Eacc$ is defined earlier in this
  section.
 
  Let $\ket{\rhobar}\in \sym^N(\hilbert_P\otimes \hilbert_S)$
  be an arbitrary purification of $ \frac
  1{N!}\sum_{\pi\in \symgrp N}\linebreak \pi_{S^N} \rho_{S^N} \pi_{S^N}^\dagger$. Since all purifications are equivalent up
  to an isometry on the purifying register, there exists an isometry
  ${W}_{P^N\rightarrow R\bar\Pi}$ such that
  \begin{equation*}
    ({W}_{P^N}\otimes
    {\id}_{S^N})\ket{\rhobar} = \frac 1{\sqrt{N!}}\sum_{\pi\in \symgrp n}
    (\id_R\otimes 
    \pi_{S^N}) \ket{\rho_{RS^N}}\otimes \ket \pi_{\bar \Pi} \enspace.
  \end{equation*}  Let $\bar U_{ P^N\rightarrow \bar RP^k}$ be the
  isometry that performs the following actions 
  unitarily on register $P^N$ of $\ket{\rhobar}$:
\begin{enumerate}
\item Apply ${W}_{P^N}$, producing registers $R$ and $\bar \Pi$.
\item From permutation $\pi\in \symgrp N$ held in register $\bar \Pi$, compute
  $t_\pi$, $\tau^\pi\in \symgrp k$ and $\bar \tau^\pi\in \symgrp n$ as in
  Remark~\ref{rem:dividepi}, i.e. such that $\Vk\cdot\pi_S = (\tau^\pi_{\hat S} \otimes
  \bar\tau^\pi_{S'})\cdot \Vpi$.
\item Apply attack $\Upi$ on register $R$, producing registers $R'$
  and $P^k$ and reorder register $P^k$ using permutation $\tau^\pi$ so
  that each $P_i$ aligns with the right sampled $S_i$.
  \label{itm:reorder}
\item Let register $\bar R$ be composed of registers $R'$, $\bar \Pi$.
  Output registers $P^k$ , $\bar R$ and register $\Pi$ containing the
  permutation $\bar \tau^\pi$ that acts on the output $S^n$ (i.e. on the
  unsampled registers).\label{step:advoutput} 
 \end{enumerate}
 From the definition of the above isometry, 
\begin{align*}
  &(\bar U_{ P^N\rightarrow \bar RP^k}\otimes \Vk)\ket{\rhobar}\\
  &= \frac 1{\sqrt{N!}}
    \sum_{\pi\in \symgrp N}(\tau^\pi_{ P^k}\otimes \tau^\pi_{ S^k} \otimes \bar\tau^\pi_{S^n})(\Upi\otimes \Vpi) 
    \ket{\rho_{RS^N}} \ket \pi_{\bar \Pi} \ket{\bar \tau^\pi}_{\Pi}
\end{align*}
Tracing out register $\bar \Pi$ from the above and using the
one-to-one correspondence between $\pi$ and $(t_\pi, \tau^\pi,
\bar\tau^\pi)$ to break the sum over $\pi$ into sums over $t$, $\tau$
and $\bar \tau$, we get
\begin{align*}
  &\frac 1{N!} \sum_{\pi\in \symgrp N} 
    [(\tau^\pi_{ P^k}\otimes \tau^\pi_{ S^k} \otimes
    \id_{R'}\otimes \bar\tau^\pi_{S^n})(\Upi \otimes \Vpi)] (\rho_{RS^N})\otimes
    \proj{\bar \tau^\pi}_{\Pi}\\
  &  =\frac 1{n!} \frac 1{k!}\frac 1{\binom Nk}\sum_{\bar \tau\in \symgrp n}  
    \bar\tau_{S^n}
    \Bigg(
    \sum_{\substack{\tau\in \symgrp k\\t\subseteq [N]: |t|=k}} 
  [(\tau_{ P^k}\otimes \tau_{ S^k})
  (U^t_R \otimes V^t_{S^N})]
  (\rho_{RS^N})\Bigg)(\bar\tau_{S^n})^\dagger\otimes
  \proj{\bar \tau^\pi}_{\Pi}
\end{align*}
Taking the partial inner product with $\ket\varphi^{\otimes k}_{P^kS^k}$ and tracing out $R'$ leaves us with
\begin{align*}
  &\frac 1{n!\binom Nk} 
  \sum_{\bar \tau\in \symgrp n}  \bar\tau_{S^n}\Bigg(\sum_{t} \trace[R']{\bra\varphi^{\otimes k}_{
  P^k S^k}\cdot[
  U^t_R \otimes V^t_{S^N}]
  (\rho_{RS^N})\cdot\ket\varphi^{\otimes k}_{
  P^k S^k}}\Bigg)( \bar\tau_{S^n})^\dagger\otimes
    \proj{\bar \tau^\pi}_{\Pi}\\
  &\qquad= \frac 1{n!}\sum_{\bar\tau\in \symgrp n}
  \bar \tau_{S^n} \Eacc(\rho_{RS^N})\bar\tau_{S^n}^\dagger\otimes
    \proj{\bar \tau^\pi}_{\Pi}
\end{align*}
where the sum over $\tau$ disappeared because $\ket\varphi^{\otimes
  k}_{ P^k S^k}$ is invariant under permutation.  Then $\barEacc$
defined as
    \begin{align*}
      \barEacc(\rhobar):= 
      \trace[\bar R]{\bra{\varphi}^{\otimes k}_{ P^k S^k}\cdot
        [\bar U_{P^N}\otimes V^{[k]}_{S^N}](\rhobar)
        \cdot\ket{\varphi}^{\otimes 
          k}_{P^kS^k}} \enspace.
  \end{align*}
  satisfies~(\ref{eq:4rccr248cr9284cr2938}).  \qed
\end{proof}

\begin{lemma}
  Protocol \purifbased from Fig.~\ref{fig:samplingprot} satisfies the
  second criterion of Definition~\ref{def:perminv}.
\end{lemma}
\begin{proof}
  We need to show that for any $\epsilon>0$, $\|\barEacc(\proj\theta_{
    P^NS^N}^{\otimes N})\|_1 \leq \exp(-\Omega(N))$ whenever
  $F(\theta_{S}, \varphi_{ S})^2< 1-\epsilon$ where
 \begin{align*}
   \barEacc(\rhobar):= 
   \trace[\bar R]{\bra{\varphi}^{\otimes k}_{ P^k S^k}\cdot
     [\bar U_{P^N}\otimes V^{[k]}_{S^N}](\rhobar)
     \cdot\ket{\varphi}^{\otimes 
       k}_{P^kS^k}}\enspace.
    \end{align*} 
    The proof is based on the simple observation that the isometry
    $\bar U$ that maximizes the probability of observing
    $\ket{\varphi}^{\otimes k}$ on registers $P^k S^k$ is the one that
    matches the fidelity with ${\varphi}^{\otimes k}$ by the fact that
    the fidelity is monotonous. Therefore it holds that, since the
    fidelity is multiplicative for product states,
    \begin{equation*}
      \|\barEacc\left(\proj\theta_{ P^NS^N}^{\otimes 
          N}\right)\|_1 \leq F(\theta^{\otimes k}_{ S^k},
      \varphi^{\otimes k}_{ S^k})^2\leq
      (1-\epsilon)^{2k}\leq \exp(-2 \epsilon k)
    \end{equation*}
    whenever $F(\theta_{ S}, \varphi_{ S})^2< 1-\epsilon$
    .\qed
\end{proof}

The third criterion of Definition~\ref{def:perminv} follows trivially
from the observation that neither $\Eacc$ nor $\barEacc$ acts on the
unsampled qubits other than by rearranging them.

\subsection{Proof of Proposition~\ref{prop:EPRsampling}}
\label{sec:proofeprsampl}

As in Section~\ref{sec:purifscheme}, let us establish that the
protocol satisfies the each criterion of Definition~\ref{def:perminv}.
\begin{lemma}[First criterion]
\label{lem:firstcritepr}
  Let $\Eacc$ be the output of the
  sampling protocol \eprlocc from Fig.~\ref{fig:EPRsampling}.  For any
  $\rho_{RS^N}\in \density{\hilbert_R\otimes \hilbert_S^{\otimes N}}$ there
  exists $\barEacc$ such that
    \begin{equation}
      \label{eq:f094hf9hg02hg3}
      \frac 1{n!}\sum_{\pi\in \symgrp n} \proj\pi_\Pi\otimes
      \pi_{S^n} \Eacc(\rho_{RS^N})
      \pi_{S^n}^\dagger=
      \barEacc(\rhobar)
    \end{equation}
    for some symmetric purification $\ket{\rhobar}$ of $
    \frac 1{N!}\sum_{\pi\in \symgrp N} \pi_{S^N} \rho_{S^N} \pi_{S^N}^\dagger$.
\end{lemma}
\begin{proof}
  Recall the linear operator $\Vt$ from Section~\ref{sec:purifscheme}
  that maps $S_t$ to $S^k$ and $S_{\bar t}$ to $S^n$ (where $S^k$ is
  understood to represent the last $k$ registers). The \cptn map
  $\Eacc$ that models the action of the protocol on the state
  $\rho_{RS^N}$ when Sam accepts can be represented as 
  \begin{align*}
    2^{-k}\binom Nk^{-1}\sum_{t,c,x}
    \trace[RS^k]{ (E^{t,c}_x\otimes \mathbb P_{S^k}^{x,c}) \Vt \rho_{RS^N} \Vt}
  \end{align*}
  where the sum is over $t\subset[N]$ such that $|t|=k$, $c\in
  \{+,\times \}^k$ and $x\in\bool^k$ and where, for $t$ and $c$ sent
  by Sam, $E^{t,c}= \{E^{t,c}_x\}_{x\in\bool^k}$ is the POVM
  measurement on $R$ that produces $x$ and $\mathbb P_{S^k}^{x,c}:=
  H^{\otimes c}\proj xH^{\otimes c}$ is the projector onto $x$ in
  basis $c$.

  Let $\rhobar$ be an arbitrary purification of $ \frac
  1{N!}\sum_{\pi\in \symgrp N} \pi_{S^N} \rho_{S^N}
  \pi_{S^N}^\dagger$. Define the map $\barEacc$ as follows:
  \begin{enumerate}
  \item Map state $\rhobar$ to $\frac 1{N!}\sum_{\pi\in \symgrp N} \proj
    \pi_{\bar \Pi}\otimes (\id_R\otimes \pi_{S^N}) \rho_{RS^N}
    (\id_R\otimes\pi_{S^N}^\dagger)$.
  \item From permutation $\pi\in \symgrp N$ held in register $R$, compute
    $t_\pi$, $\tau^\pi\in \symgrp k$ and $\bar \tau^\pi\in \symgrp n$ as in
    Remark~\ref{rem:dividepi}.
  \item Apply $\Vk$ on $S^N$, choose $c\in\{+,\times \}^k$ at random
    and apply POVM $E^{t_\pi, c}$ on $R$ producing output
    $x$.\label{st:0f892f4q9p}
  \item Measure the sampled registers $S^k$ by projecting on $
    H^{\otimes \tau^\pi(c)}\ket{\tau^\pi(x)}_{S^k}= \tau^\pi
    H^{\otimes c}\ket{x}_{S^k}$.
  \item Output $\bar \tau^\pi$ in register $\Pi$ and register $S^n$.
  \end{enumerate}

  The output of $\barEacc$
  applied on $\rhobar$ is
  \begin{align*}
    &
        \frac {2^{-k}}{N!}\sum_{\pi, c,x} \tr_{RS^k}
    \Big((E^{t_\pi,c}_x\otimes \tau^\pi _{ S^k}\mathbb P^{x,c}_{
          S^k}(\tau^\pi _{ S^k})^\dagger)\cdot
        [\Vk\pi_{S^N}](\rho_{RS^N})\Big)\otimes \proj{\bar
          \tau^\pi}_\Pi
\\
    &=
         \frac
         {2^{-k}}{N!}\sum_{\pi,c,x}\bar\tau^\pi_{S^n}\tr_{RS^k}\Big
         (E^{t_\pi,c}_x\otimes \mathbb P^{x,c}_{S^k}) [\Vpi]
         (\rho_{RS^N}) \Big)\bar\tau^\pi_{S^n}\otimes \proj{\bar
           \tau^\pi}_\Pi
\\
    &=
       \begin{aligned}[t]
         \frac {2^{-k}}{n!} \binom Nk^{-1}\sum_{\bar\tau^\pi\in
           \symgrp
           n}[\bar\tau^\pi_{S^n}]\Bigg(&\sum_{t,c,x}\tr_{RS^k}\Big(
         (E^{t,c}_x\otimes \mathbb P^{x,c}_{S^k}) [V^t_{S^N}](
         \rho_{RS^N}) \Big)\Bigg)\otimes \proj{\bar
           \tau^\pi}_\Pi
       \end{aligned}
\\
    &= \frac  {1}{n!}\sum_{\bar\tau^\pi\in
       \symgrp n}\bar\tau^\pi_{S^n} \Eacc(\rho_{RS^N})\bar\tau^\pi_{S^n}\otimes 
       \proj{\bar \tau^\pi}_\Pi 
  \end{align*}
  where the second equality uses Remark~\ref{rem:dividepi}.\qed
\end{proof}

\begin{lemma}[Second criterion]
  Let $\barEacc$ be as in the proof of
  Lemma~\ref{lem:firstcritepr}. For any $\epsilon>0$,
  $\|\barEacc(\proj\theta_{P^NS^N}^{\otimes N})\|_1 \leq
  \exp(-\Omega(N))$ whenever $F(\theta_{ S}, \varphi_{ S})^2<
  1-\epsilon$
\end{lemma}
\begin{proof}
  For any $c\in\{+,\times \}^k$, let $\bar E^c_x$ be the POVM element
  on $P^N$ that gives the probability of $x$ being outputted in
  step~\ref{st:0f892f4q9p} of $\barEacc$ when $c$ is chosen in the
  same step. In essence, $\bar E^c_x$ is to $\barEacc$ what $E^{t_\pi,
    c}$ is to $\Eacc$; it gives the probability of observing $x$ when
  the following measurement is done on $P^N$: produce registers
  $\bar\Pi R$ from $P^N$, measure $\pi$ from register $\bar \Pi$,
  compute the corresponding sample $t_\pi$, and apply the measurement
  corresponding to POVM $E^{t_\pi, c}$.

  Using these POVM operators $\bar E^c_x$, we can express the norm we
  wish to upper-bound as
  \begin{equation} \label{eq:f0q40gffffadf}
    \|\barEacc(\proj\theta_{P^NS^N}^{\otimes N})\|_1=
        {2^{-k}}\sum_{c,x}\trace{(\bar E_x^c \otimes \mathbb P^{x,c}_{S^k} \otimes \id_{S^n})
        \proj\theta_{P^NS^N}^{\otimes N}}
  \end{equation}
  where $\mathbb P^{x,c}_{S^k}$ is the projector onto $x$ in basis
  $c$. Note that the right-hand side of~(\ref{eq:f0q40gffffadf}) can be
  interpreted as the probability of guessing the outcome of measuring
  register $ S^k$ in a known but random basis $c$ by observing the
  reduced operator of register $P^N$. We now analyze this guessing
  probability to provide an upper-bound on~(\ref{eq:f0q40gffffadf}).

  Since each measurement on $ S^k$ is independent of each other and
  since the joint state is in an i.i.d. form, the probability of Paul
  guessing outcome $x$ is of the form $\gamma^k$ where $\gamma$
  corresponds to the probability of guessing a single bit of $x$. This
  probability is given by the expression
  \begin{equation*}
    \gamma= \frac 12 \Pr(\text{guess }X\mid C=+) + \frac 12
    \Pr(\text{guess }X\mid C=\times )
  \end{equation*}
  We show that at least one of the above conditional term is bounded
  above by a constant strictly smaller than $1$ when $F(\theta_S,
  \varphi_S)<1-\epsilon$, which means that $\gamma^k$ is negligible in
  $k$.

  The maximum probability of guessing $X$ when $C=+$ is given by the
  probability of distinguishing states
  \begin{equation*}
    \ket{\theta^0_P}= (\id_p\otimes \bra 0_S)\ket{\theta_{PS}}\text{
      and } \ket{\theta^1_P}= (\id_p\otimes \bra 1_S)\ket{\theta_{PS}}
  \end{equation*}
  and the same holds when $C=\times$ for similarly defined
  $\ket{\theta^+_P}$ and $\ket{\theta^-_P}$. Let
  \begin{equation*}
    \sqrt{\lambda_0} \ket{f_0}_P\ket{e_0}_S+
    \sqrt{\lambda_1} \ket{f_1}_P\ket{e_1}_S
  \end{equation*}
  be the Schmidt decomposition of $\ket{\theta_{PS}}$ and consider the
  quantity
  \begin{align*}
    \label{eq:fqoiwejqfpof4h}
    &\left|\braket{\theta^0_P}{\theta^1_P}\right|+
      \left|\braket{\theta^+_P}{\theta^-_P}\right|
      \geq\left|\braket{\theta^0_P}{\theta^1_P}+
      \braket{\theta^+_P}{\theta^-_P}\right|\\
    &= \left|\bra{\theta_{PS}}(\id_P\otimes \ketbra 01_S) \ket{\theta_{PS}}
      + \bra{\theta_{PS}}(\id_P\otimes \ketbra +-_S)
      \ket{\theta_{PS}}\right| \\ 
    &= \frac 12\left|\bra{\theta_{PS}}(\id_P\otimes H_S)
      \ket{\theta_{PS}} \right|=  \frac 12\left|{\lambda_0} \bra{e_0}_S H_S\ket{e_0}_S+
      {\lambda_1} \bra{e_1}_S H_S\ket{e_1}_S\right|\\&=  \frac 12\left|{\lambda_0}-
      {\lambda_1}\right|
  \end{align*}
  where $H_S= \begin{pmatrix}1&1\\1&-1\end{pmatrix}$, the only
  inequality above is the triangle inequality and the last equality
  follows from the fact that $\bra{e_0}_S H_S\ket{e_0}_S =
  -\bra{e_1}_S H_S\ket{e_1}_S$ for any two orthogonal vectors
  $\ket{e_0}_S$ and $\ket{e_1}_S$.  The last term from the above
  equation can be bounded above by $\epsilon$ since
  \begin{align*}
    \left|{\lambda_0}-{\lambda_1}\right|
    = \left|{\lambda_0}-\frac 12\right| +
    \left|{\lambda_1}-\frac 12\right|
    = \left\| \theta_S - \frac{\id_S}2 \right\|_1
    \geq 2(1-F(\theta_S, \frac {\id_S}2))\geq 2\epsilon
  \end{align*}

  Suppose that $\left|\braket{\theta^0_P}{\theta^1_P}\right|\geq
  \epsilon/2$ (otherwise,
  $\left|\braket{\theta^+_P}{\theta^-_P}\right|\geq \epsilon/2$ and
  the same argument holds for those two states), this means that Paul
  cannot distinguish between the two reduced states $\ket{\theta^0_P}$
  and $\ket{\theta^1_P}$ with probability better than one minus some
  constant (that depends on $\epsilon$). We conclude that $\gamma$ is
  bounded above by a constant strictly less than $1$ and that the
  probability $\gamma^k$ of guessing all measurement outcomes
  correctly declines exponentially fast in $k$.\qed

\end{proof}

The third criterion of Definition~\ref{def:perminv} follows trivially
from the observation that neither $\Eacc$ nor $\barEacc$ acts on
the unsampled qubits other than by relabeling them.


%% file: app-proofs2.tex
\begin{proof}[Proposition \ref{prop:leqequivpostsel}]
  Let's start with the easier direction of the proof. Let  $
  \ket{\sigma_{R_1Q}}$ be a purification of $\sigma_Q$, let
  $\ket{\rho_{R_2Q}}$ be a purification of $\rho_Q$ and let $A_{R_1\rightarrow R_2}$ be
  as in (\ref{eq:post-selection}). Then by Remark~\ref{rem:fq09h4g},
  $\rho_Q$ is equal to
  \begin{align*}
     \trace[R_2]{\rho_{R_2Q}}=
    c\cdot\trace[R_1]{(A_{R_1\rightarrow R_2}^\dagger
    A_{R_1\rightarrow R_2}\otimes 
    \id_Q)\sigma_{R_1Q}}\leq
    c\cdot\trace[R_1]{\sigma_{R_1Q}}= c\cdot\sigma_Q  \enspace.
  \end{align*}

  For the other direction, write $\sigma_Q$ as $ \sigma_Q= \frac 1c
  (\rho_Q+\tilde\sigma_Q) \text{ where } \tilde \sigma_Q := c\cdot
  \sigma_Q-\rho_Q \geq 0\enspace .  $ Let $\ket{\rho_{R_2Q}}$ be an
  arbitrary purification of $\rho_Q$ and let $\ket{\tilde
    \sigma_{R_2Q}}$ be a purification of $\tilde\sigma_Q$ that lives
  in the same space. Then consider the following purification of
  $\sigma_Q$: $ \ket{\sigma_{R'R_2Q}}:= \sqrt{\frac 1c}(\ket{0}_{R'}
  \ket{\rho_{R_2Q}} + \ket 1_{R'}\ket{\tilde\sigma_{R_2Q}}) $. Let
  $\ket{\sigma_{R_1Q}}$ be an arbitrary purification of $\sigma_Q$ and
  let $A_{R_1\rightarrow R_2}:= (\bra 0_{R'}\otimes
  \id_{R_2})V_{R_1\rightarrow R'R_2}$ where $V_{R_1\rightarrow R'R_2}$
  is an isometry that maps $\ket{\sigma_{R_1Q}}$ to
  $\ket{\sigma_{R'R_2Q}}$. Then
  \begin{equation*}
    (A_{R_1\rightarrow R_2}\otimes \id_Q)\ket{\sigma_{R_1Q}}= (\bra 0_{R'}\otimes
    \id_R)\ket{\sigma_{R'R_2Q}} = \sqrt {\frac 1c} \ket{\rho_{R_2Q}}
    \enspace .
  \end{equation*}
  \qed
\end{proof}

\begin{proof}[Proposition \ref{prop:smallnumbterms}]
  It suffices to show that $\bra a \left(
  |\mathcal J| \cdot\rho^{mix}-\rho \right) \ket a\geq 0$ for any
$\ket a\in \hilbert$. Consider the following chain of (in)equalities:
\begin{multline*}
  |\mathcal J|\bra a \rho^{mix} \ket a 
  = |\mathcal J|\bra a\left(\sum_{i\in \mathcal J}
    \proj{\psi_i}\right)\ket a 
  = |\mathcal J|\sum_{i\in \mathcal J} 
    \left|\braket{a}{\psi_i}\right|^2 \\
  \geq \left|\sum_{i\in \mathcal J} 
    \braket{a}{\psi_i}\right|^2 
  =   \left(\sum_{i\in \mathcal J} 
    \braket{a}{\psi_i}\right)\left(\sum_{j\in \mathcal J} 
    \braket{\psi_j}a\right) 
  = \bra a\left(\sum_{i,j\in \mathcal J}
     \ketbra{\psi_i}{\psi_j}\right) \ket a
  = \bra a\rho \ket a
\end{multline*}
where the only inequality above follows from the Cauchy-Schwarz
inequality ``$|\braket \varphi\psi|^2 \leq \braket \varphi\varphi
\braket \psi\psi$'' with $\ket \varphi = \sum_{i\in \mathcal J} \ket i$
and $\ket \psi= \sum_{i\in\mathcal J}\braket a{\psi_i} \ket
i$. This completes the proof.\qed
\end{proof}

\begin{proof}[Lemma \ref{lem:puriftauisideal}]
  Observe that
  \begin{align*}
    \trace{\mathbb P^{r, \ket\nu} \proj\theta^{\otimes
    n} }=\Pr[ wt(X_\theta)\leq r]=\Pr[ wt(X_\theta) - \epsilon n\leq
    \alpha n] 
  \end{align*}
  where $X_\theta$ is a random variable obtained by measuring $n$
  copies of $\ket \theta$ with observables $M_0= \proj\nu$ and $M_1=
  \id-\proj \nu$ and where $wt(\cdot)$ is the Hamming weight function,
  i.e. the number of ones. Since $X_\theta$ consists of $n$
  i.i.d. Bernoulli trials with parameter $1-F(\nu, \theta)^2\leq
  \epsilon$, Hoeffding's inequality allows us to lower-bound the above
  quantity: $ \trace{\mathbb P^{r, \ket\nu}\, \proj\theta^{\otimes n}
  }\geq 1-\exp(-2\alpha^2 n)$. \qed
\end{proof}

\begin{proof}[Lemma \ref{lem:unpermuteideal}]
  Let $r=\epsilon n$. We need to show that if $\bar
  \sigma_{S^n}:=\frac 1{n!}\sum_{\pi\in \symgrp n} \pi_{S^n}
  \sigma_{S^n}\pi^\dagger_{S^n}$ has a purification in
  $\hilbert_{R}\otimes \Delta_{r}(\ket \varphi^{\otimes
    n}_{{P^n}{S^n}})$ for some register $R$, then $\sigma_{S^n}$ also
  has a purification in $\hilbert_{R}\otimes \Delta_{r}(\ket
  \varphi^{\otimes n}_{{P^n}{S^n}})$.  Let $\ket{\bar
    \sigma_{R{P^n}{S^n}}}\in \hilbert_{R}\otimes \Delta_{r}(\ket
  \varphi^{\otimes n}_{{P^n}{S^n}})$ be the purification of $\bar
  \sigma_{S^n}$ that exists by assumption and let
  $\sum_ip_i\proj{i_{S^n}}$ be the spectral decomposition of
  $\sigma_{S^n}$. Define the pure state
  \begin{equation*}
    \ket{ \bar \sigma_{\Pi {P^n}{S^n}}}= \sqrt{\frac 1{n!}} \sum_{\pi\in
      \symgrp n} \ket\pi_\Pi 
    \otimes \left(\sum_i \sqrt{p_i} \ket{i_{P^n}} \otimes \pi_{S^n}\ket
      {i_{S^n}}\right) 
  \end{equation*}
  where $\{\ket{i_{P^n}}\}_i$ is an orthonormal basis of
  $\hilbert_{P^n}$. Note that this state is a purification of
  $\bar\sigma_{S^n}$, so there exists an isometry $V_{\Pi
    {P^n}\rightarrow R{P^n}}$ such that $V_{\Pi {P^n}\rightarrow R{P^n}}
  \ket{\bar \sigma_{\Pi {P^n}{S^n}}} = \ket{\bar \sigma_{R{P^n}{S^n}}}\in
  \hilbert_{R}\otimes \Delta_{r}(\ket \varphi^{\otimes
    n}_{{P^n}{S^n}})$.
  We can express $\ket{\bar \sigma_{R {P^n}{S^n}}} $ as:
  \begin{align*}
    &\ket{\bar \sigma_{R{P^n}{S^n}}}
    = (V_{\Pi {P^n}\rightarrow R{P^n}}\otimes \id_{S^n})
      \ket{\sigma_{\Pi {P^n}{S^n}}} \\ 
    &= \sum_{\pi,i} \sqrt{\frac {p_i}{n!}} V_{\Pi {P^n}\rightarrow
      R{P^n}}\ket\pi_\Pi 
      \ket{i_{P^n}} \otimes \pi_{S^n}\ket {i_{S^n}}=\sum_{\pi,i} \sqrt{\frac {p_i}{n!}} \ket{\xi_{\pi,i}}_{R{P^n}} \otimes
      \pi_{S^n}\ket {i_{S^n}}
  \end{align*}
  where the vectors $\ket{\xi_{\pi,i}}_{R{P^n}}:= V_{\Pi
    {P^n}\rightarrow R{P^n}}\ket\pi_{\Pi} \ket{i_{P^n}}$ are orthogonal
  to each other. Then by acting on this state with  an isometry that
  extracts $\pi$ from registers $R{P^n}$ and that undoes $\pi$ on
  registers $P^n$ and $S^n$, we get 
    \begin{equation*}\sum_{\pi,i} \sqrt{\frac {p_i}{n!}}
      (\id_R\otimes\pi_{P^n}^{-1})\ket{\xi_{\pi,i}}_{R{P^n}} \otimes 
      \ket {i_{S^n}}
  \end{equation*}
  Note that both before and after this isometry is applied, the state
  of registers ${P^n}$ and ${S^n}$ has support in $\Delta_{r}(\ket
  \varphi^{\otimes n}_{{P^n}{S^n}})$ because this subspace is
  invariant under permutation of these registers. The proof is then
  completed since the above state is a purification of $\sigma_{S^n}$
  that lies in $\hilbert_R\otimes \Delta_{r}(\ket \varphi^{\otimes
    n}_{{P^n}{S^n}})$.
  \qed
\end{proof}
